\documentclass[12pt]{article}

\usepackage{graphicx}
\usepackage{amsmath}
\usepackage{amssymb}
\usepackage{amsthm}

\usepackage{geometry}
\usepackage[colorlinks = true]{hyperref}

\hypersetup{
  pdftitle = {Gapped and gapless phases of frustration-free spin-1/2 chains},
  pdfauthor = {Sergey Bravyi, David Gosset}
}

\usepackage{xcolor}
\definecolor{darkred}  {rgb}{0.5,0,0}
\definecolor{darkblue} {rgb}{0,0,0.5}
\definecolor{darkgreen}{rgb}{0,0.5,0}
\hypersetup{
  urlcolor   = blue,         
  linkcolor  = darkblue,     
  citecolor  = darkgreen,    
  filecolor  = darkred       
}

\usepackage{subfig}

\newcommand{\be}{\begin{equation}}
\newcommand{\ee}{\end{equation}}
\newcommand{\ba}{\begin{array}}
\newcommand{\ea}{\end{array}}
\newcommand{\bea}{\begin{eqnarray}}
\newcommand{\eea}{\end{eqnarray}}

\newcommand{\calB}{{\cal B }}

\newcommand{\calG}{{\cal G }}

\newcommand{\calA}{{\cal A }}

\newcommand{\CC}{\mathbb{C}}

\newcommand{\RR}{\mathbb{R}}

\newcommand{\trace}{{\mathrm{Tr}}}

\newtheorem{prop}{Proposition}
\newtheorem{Lemma}{Lemma}

\newtheorem{theorem}{Theorem}
\newtheorem{regionexclusion}{Region Exclusion Lemma}
\newtheorem*{gappedphasethm}{Gapped phase theorem}
\newtheorem*{gaplessphasethm}{Gapless phase theorem}

\title{Gapped and gapless phases of frustration-free spin-$\frac{1}{2}$ chains}
\author{Sergey Bravyi\footnote{IBM T.J. Watson Research Center, Yorktown Heights, NY}\and David Gosset\footnote{Walter Burke Institute for Theoretical Physics and Institute for Quantum Information and Matter, California Institute of Technology} \footnote{Department of Combinatorics \& Optimization and Institute for Quantum Computing, University of Waterloo}}
\date{June 1, 2015}
\begin{document}
\maketitle

\begin{abstract}
We consider a family of translation-invariant quantum spin chains with nearest-neighbor interactions and derive necessary and sufficient conditions for these systems to be gapped in the thermodynamic limit. More precisely, let $\psi$ be an arbitrary two-qubit state. We consider a chain of $n$ qubits with open boundary conditions and Hamiltonian $H_n(\psi)$ which is defined as the sum of rank-$1$ projectors onto $\psi$ applied to consecutive pairs of qubits. We show that the spectral gap of  $H_n(\psi)$ is upper bounded by $1/(n-1)$ if the eigenvalues of a certain $2\times 2$ matrix simply related to $\psi$ have equal non-zero absolute value. Otherwise, the spectral gap
is lower bounded by a positive constant independent of $n$ (depending only on $\psi$).  A key ingredient in the proof is a new operator inequality for the ground space projector which expresses a monotonicity under the partial trace. This monotonicity property appears to be very general and might be interesting in its own right. As an extension of our main result, we obtain a complete classification of gapped and gapless phases of frustration-free translation-invariant spin-$1/2$ chains with nearest-neighbor interactions.
\end{abstract}

\newpage
\tableofcontents


\section{Introduction}
Many properties of quantum spin chains depend crucially on whether the Hamiltonian
is gapped or gapless in the thermodynamic limit. Ground states of gapped Hamiltonians
are weakly entangled, as quantified by
the entanglement area law~\cite{Hastings2007area,Arad2013area,Brandao2013area},
and exhibit an exponential decay of correlation functions~\cite{Hastings2006spectral}.
For such systems the ground energy and the ground state itself can be efficiently computed
using algorithms based on Matrix Product States~\cite{Vidal2003efficient,perez2006matrix,Landau2013polynomial,Chubb2015}.
On the other hand, ground states of gapless spin chains
can exhibit drastic violations of the entanglement area 
law~\cite{Gottesman09,Irani09,Bravyi2012criticality,Movassagh2015power},
and computing the ground energy can be quantum-NP 
hard~\cite{Aharonov2009power,Hallgren2013local}.
Spin chain models studied in physics usually become gapless along quantum phase transition lines 
separating distinct gapped phases~\cite{Sachdev11}. Deciding whether a given family of  Hamiltonians is gapped or gapless in the thermodynamic limit 
is therefore a fundamental problem.

In this paper we provide a complete solution of this problem for a class of translation-invariant
chains of qubits with nearest-neighbor interactions. 
Let $\psi\in \CC^2\otimes \CC^2$ be a fixed two-qubit state with $\|\psi\|=1$. 
Consider a chain of $n$ qubits with open boundary conditions and define a Hamiltonian
\begin{equation}
\label{H}
H_n(\psi)  =\sum_{i=1}^{n-1}|\psi\rangle\langle\psi|_{i,i+1}.
\end{equation}
Here each term is a rank-$1$  projector onto $\psi$  applied to a consecutive pair of qubits. We shall refer to $\psi$ as the {\em forbidden state} since the Hamiltonian penalizes adjacent qubits
for being in the state $\psi$. 
 As we will see in Section \ref{sec:ground space}, the Hamiltonian $H_n(\psi)$ is frustration-free for any choice of $\psi$,
that is, ground states of $H_n(\psi)$ are zero eigenvectors of each individual projector $|\psi\rangle\langle\psi|_{i,i+1}$ and the ground energy of $H_n(\psi)$
is zero. Futhermore, the ground state degeneracy of $H_n(\psi)$ is equal to $n+1$ for almost all choices of $\psi$.

We are interested in the spectral gap separating the ground states
and the excited states of $H_n(\psi)$ or, equivalently, the smallest non-zero eigenvalue
of $H_n(\psi)$. To state our results, define a $2\times 2$ matrix
\begin{equation}
\label{Tpsi}
T_\psi=\left( \begin{array}{rr}
\langle \psi |0,1\rangle & \langle \psi|1,1\rangle \\
-\langle \psi|0,0\rangle & -\langle \psi|1,0\rangle \\
\end{array}
\right).
\end{equation}
Here $|0\rangle,|1\rangle$ is the standard basis of $\CC^2$.
As we will see, 
the matrix $T_\psi$ is crucial for understanding the structure of the ground space of $H_n(\psi)$. In this paper we prove that the eigenvalues of $T_\psi$ determine if $H_n(\psi)$ is gapped or gapless\footnote{Although our definition of 
the matrix
$T_\psi$ is basis-dependent, eigenvalues of $T_\psi$ are invariant under global $SU(2)$ rotations.
In particular, one can check that  a transformation $\psi \to (U\otimes U) \psi$ with a single-qubit unitary
operator $U$ maps
$T_\psi$ to $(\det{U})^{-1} \cdot U T_\psi U^\dag$.}.
Our main result is the following.
\begin{theorem}
\label{thm:main}
Let $\psi$ be an arbitrary two-qubit state. 
Suppose the eigenvalues of $T_\psi$ have equal non-zero 
absolute value. Then the spectral gap of $H_n(\psi)$ is at most $1/(n-1)$.
Otherwise the spectral gap of $H_n(\psi)$ is lower bounded by a positive constant independent of $n$, which depends only on the forbidden state $\psi$.
\end{theorem}

We now motivate our choice of the model Eq.~(\ref{H}),
highlight previous work on related models, 
and provide some intuition for why the eigenvalues of $T_\psi$ appear in the statement of the
theorem.
An informal sketch of the proof is provided in Section~\ref{sketch}.
Below we write $\gamma(\psi,n)$ for the spectral gap of $H_n(\psi)$.

The family of Hamiltonians defined in Eq.~(\ref{H}) includes some
well-known quantum models as special cases. 
For example, choosing $\psi$ proportional to  $|0,1\rangle-|1,0\rangle$
(the singlet state) one can easily check that $H_n(\psi)$ coincides with the
ferromagnetic Heisenberg chain up to an overall energy shift. 
For this model $H_n(\psi)$ has spectral gap $\gamma(\psi,n)=1-\cos{(\pi/n)}$ which decays 
as $n^{-2}$ for large $n$~\cite{KN}.  Note that in this case $T_\psi$ is proportional to the identity matrix,
so Theorem~\ref{thm:main} gives an upper bound $\gamma(\psi,n)\le 1/(n-1)$.
Koma and Nachtergaele studied a one-parameter deformation
of the Heisenberg chain known as the ferromagnetic XXZ chain with kink boundary conditions~\cite{KN}.
In this example $\psi$ is proportional to $|0,1\rangle-q|1,0\rangle$ for $q>0$ 
and the spectral gap of $H_n(\psi)$ is given by 
\[
\gamma(\psi,n)=1-2(q+q^{-1})^{-1} \cos{(\pi/n)}
\]
for all $n\ge 2$, see~\cite{KN} for details.
One can check that $T_\psi$ is a diagonal matrix with eigenvalues $\mu_1=(1+q^2)^{-1/2}$ and
$\mu_2=q(1+q^2)^{-1/2}$. It follows that $|\mu_1|\ne |\mu_2|$ for any $q\ne 1$
and Theorem~\ref{thm:main} asserts that $H_n(\psi)$ has a constant spectral gap. We note that in the two special cases considered above the Hamiltonian has a symmetry 
which enables an exact computation of the  spectral gap.  
Such symmetries are not available for a general state $\psi$.

The exact results summarized above may suggest that the Hamiltonian
$H_n(\psi)$ is gapless if $\psi$ is a maximally entangled state and
gapped otherwise. Theorem~\ref{thm:main} demonstrates that this  naive intuition 
is wrong. Indeed, choose $\psi$ proportional to $\sqrt{1-p} |0,0\rangle + \sqrt{p}|1,1\rangle$
for some $0<p<1$.
 Then the matrix $T_\psi$ has eigenvalues $\pm i\sqrt{p(1-p)}$ and Theorem~\ref{thm:main} implies
that $H_n(\psi)$ is gapless for all  $p$ as above. 

As a simple application of Theorem~\ref{thm:main}, we now map out the phase
diagram of $H_n(\psi)$ restricted to  the subset of {\em real} states $\psi\in \RR^2\otimes \RR^2$.  
Using the Schmidt decomposition 
any real two-qubit state can be  written as
\begin{equation}
|\psi_\pm\rangle=R(\theta_1)\otimes R(\theta_2) \left[ \sqrt{1-p} |0,0\rangle \pm \sqrt{p} |1,1\rangle \right],
\quad  \quad R(\theta)\equiv \left( \begin{array}{cc}
\cos{(\theta)} & \sin{(\theta)} \\
-\sin{(\theta)} & \cos{(\theta)} \\ 
\end{array} \right),
\label{realstates}
\end{equation}
for some $0\le p\le 1/2$ and  $\theta_i\in [0,\pi]$.
Since the spectrum of $H_n(\psi_\pm)$ is invariant under a simultaneous rotation of all qubits, 
the spectral gap depends only on two parameters $\theta_2-\theta_1$ and $p$.
One can easily check that the eigenvalues of $T_{\psi_+}$  have equal non-zero magnitude iff
\[
p>0 \quad \mbox{and} \quad \sin^2{(\theta_2-\theta_1)} \le \frac4{2+(p(1-p))^{-1/2}}.
\] 

On the other hand, the eigenvalues of $T_{\psi_-}$ have equal non-zero magnitude iff either $p=1/2$, or $\sin(\theta_2-\theta_1)=0$ and $0<p<1$.  These conditions determine the gapless phase of the model for the special case of real states $\psi$. The gapped and gapless regions for $\psi_+$ as a function of  $p,\theta_2-\theta_1$ are shown in Fig.~\ref{fig:real}.
A surprising feature is that the gapless phase occupies a finite volume in the parameter space.
In contrast, most of the models studied in physics only become gapless along phase transition lines which have zero measure in the parameter space.
\begin{figure}[h]
\centerline{\includegraphics[trim=10mm 64mm 21mm 70mm,clip,scale=.4]{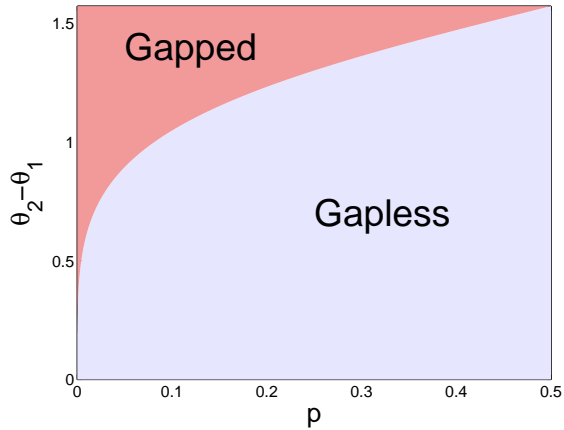}}
\caption{(Color Online) Phase diagram of the unfrustrated qubit chain 
$H_n(\psi)$  where $\psi$ has real amplitudes. We use the parameterization
$|\psi_\pm \rangle=R(\theta_1)\otimes R(\theta_2) \left[ \sqrt{1-p} |0,0\rangle \pm \sqrt{p} |1,1\rangle \right]$ and we show the gapped and gapless phases for $\psi_+$ as a function of $\theta_2-\theta_1\in [0,\pi/2]$ and $p\in [0,1/2]$. The phase diagram
is symmetric under flipping the sign of $\theta_2-\theta_1$ and under the transformation $\theta_2-\theta_1 \rightarrow \pi-(\theta_2-\theta_1)$. The sector corresponding to  $\psi_-$ is not shown since it has a simple description: $\psi_-$ is in the gapless phase iff either
$p=1/2$, or $\sin(\theta_2-\theta_1)=0$ and $0<p<1$.
\label{fig:real}
}
\end{figure}

While it is possible to construct frustration-free translation-invariant Hamiltonians on qubits which are composed of projectors of rank $2$ or $3$, one can show that there are only a handful of such examples. In the Appendix we describe them and for each we determine if the system is gapped or gapless. Taken together with our main result, this gives a complete classificaction of gapped and gapless phases for frustration-free translation-invariant qubit chains with nearest-neighbor interactions. Note that the restriction to Hamiltonians which are sums of projectors is without loss of generality.\footnote{Suppose instead we consider a frustration-free qubit chain $H=\sum_{i=1}^{n-1} h_{i,i+1}$ where $h$ has smallest eigenvalue zero (which can be arranged by adding a constant times the identity). Since $H$ is frustration-free, it has the same null space as $H'=\sum_{i=1}^{n-1} \Pi_{i,i+1}$ where $\Pi$ projects onto the range of $h$. Using this fact and the inequality $c H'\leq H\leq \|h\| H'$, where $c$ is the smallest non-zero eigenvalue of $h$, we see that $H$ is gapped if and only if $H'$ is.}

There are several open questions related to our work.  We do not know if the gapless phase of the 
model Eq.~(\ref{H}) can be connected to some known universality class of critical spin chains
and what is the actual scaling of the spectral gap in the gapless phase.
In particular, we do not expect that the  the upper bound $1/(n-1)$ on the spectral gap in Theorem~\ref{thm:main}
is tight.  It is a challenging open problem to generalize our results to qudits, i.e., to map out the phase diagram of translation-invariant frustration-free spin chains for $d$-dimensional spins with $d\ge 3$.
A natural analogue of the Hamiltonian defined in Eq.~(\ref{H}) is
\begin{equation}
H_n(\Pi)=\sum_{j=1}^{n-1} \Pi_{j,j+1},
\label{Hnpi}
\end{equation}
where $\Pi$ is a rank-$r$ projector acting on $\CC^d\otimes \CC^d$. 
It was shown by Movassagh et al~\cite{unfrustrated} that 
such chains are frustration-free for any $\Pi$ and $n$
whenever $r\le d^2/4$. Their results also suggest that the Hamiltonian $H_n(\Pi)$ may be generically frustrated for $r>d^2/4$ and $n$ sufficiently large. On the other hand, when $r>d^2/4$, the Hamiltonian is frustration-free for certain special choices of $\Pi$ and $n$; examples include the famous AKLT model~\cite{AKLT87} (with $d=3$, $r=5$), the model based on Motzkin paths~\cite{Bravyi2012criticality} (with $d=r=3$), and the ``Product Vacua with Boundary States'' models~\cite{pvbs} (with $r=(d-1)(d+2)/2$). In general there is no efficient algorithm for testing
whether $H_n(\Pi)$ is frustration-free for a given $n$
and there are indications that this problem may be computationally hard~\cite{GottesmanIraniTilings}. It is therefore natural to focus on the case $r\le d^2/4$,
where the chain is guaranteed to be frustration-free. A next step could be to investigate the phase diagram of a chain of qutrits ($d=3$) with projectors of rank $r=1,2$.

Finally, if one moves from frustration-free one-dimensional chains to general two-dimensional systems, the problem of
distinguishing between gapped and gapless phases of translation-invariant Hamiltonians becomes undecidable~\cite{CubittGap2015} which leaves no hope for mapping out the full phase diagram of such systems. 

\subsection{Sketch of the proof}
\label{sketch}

{\em Gapless phase.} In Section \ref{sec:gapless} we consider the case when  eigenvalues of $T_\psi$ have the same  non-zero magnitude
and prove  that  the spectral gap of $H_n(\psi)$ is at most
$1/(n-1)$.
The proof uses a result of Knabe~\cite{K88} relating the spectral gap 
of $H_n(\psi)$ to that of the following Hamiltonian
\begin{equation}
\label{Hcirc}
H_n^{\circ}(\psi)= H_n(\psi)+|\psi\rangle\langle \psi|_{n,1}
\end{equation}
which describes the chain with periodic boundary conditions.  The other ingredient in the proof is a detailed understanding of the ground state degeneracy of $H_n^{\circ}(\psi)$. We will see that $H_n^{\circ}(\psi)$ is always frustration-free, but its ground state degeneracy can be smaller than that of $H_n(\psi)$. In particular, if $T_\psi^n$ is not proportional to the identity operator then 
$H_n^{\circ}(\psi)$ has a two-dimensional ground space whereas $H_n(\psi)$ has an $n+1$-dimensional ground space.
Otherwise, if $T_\psi^n \sim I$, then both Hamiltonians $H_n(\psi)$ and $H_n^{\circ}(\psi)$
have ground space degeneracy $n+1$. 

We now sketch how these two ingredients can be used to prove the stated result. For ease of presentation we focus on the example considered above, where $\psi_+$ is of the form given in Eq.~\eqref{realstates}. Recall that the spectrum of $H_n(\psi_+)$ depends only on the two parameters $\theta_2-\theta_1, p$. We can plot the ground state degeneracy of $H_n^{\circ}(\psi_+)$ as a function of these two parameters. As described in the previous paragraph, this function takes the value $2$ or $n+1$ depending on whether or not $T_{\psi_+}^n$ is proportional to the identity. The black lines in Figures \ref{fig:Tncurves}(a) and \ref{fig:Tncurves}(b) show the curves where the ground state degeneracy is equal to $n+1$, for $n=10$ and $n=50$ respectively. Everywhere else ($0<p\leq \frac{1}{2}$ and $\theta_2-\theta_1\in [0,\pi/2]$) the ground state degeneracy is $2$. For reference we show the red and blue regions from Figure \ref{fig:real}, which correspond to the gapped and gapless phases of the open boundary chain in the thermodynamic limit. 
As one might guess by looking at the Figure, the black curves become dense in the blue 
region when $n\rightarrow\infty$. If we consider a point $\psi_+$ in this blue region which does not sit directly on one of the black curves then the eigenvalue gap of $H_n^\circ (\psi_+)$ is equal to its third smallest eigenvalue. However as $n\rightarrow\infty$ this point $\psi_+$ becomes arbitrarily close to a black curve, where the Hamiltonian has ground state degeneracy $n+1$ and third smallest eigenvalue equal to zero. Using a bound on its derivative one can show that as a result the third eigenvalue of $H_n^\circ (\psi_+)$ takes arbitrarily small values as $n\rightarrow\infty$. Finally, Knabe's result implies that this can occur only if the spectral gap of $H_n(\psi_+)$ is at most $1/(n-1)$.  This argument has to be modified slightly for states $\psi_+$ which, for some $n$, lie directly on one of the black curves and for general (complex) states $\psi$.

\begin{figure}[h]
\subfloat[(a)][]{
\includegraphics[trim=10mm 64mm 21mm 70mm,clip,scale=.4]{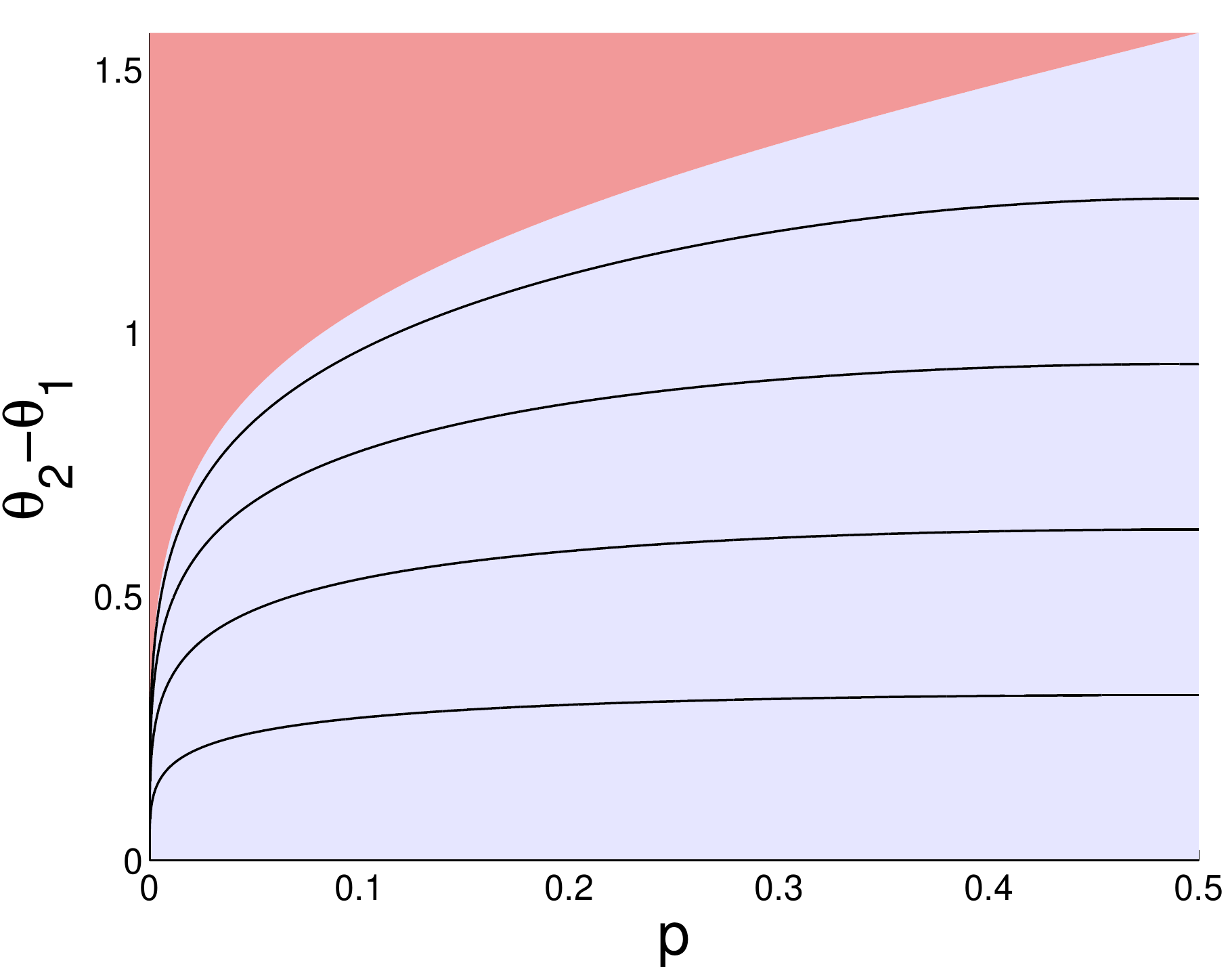}
}
\hspace{0.9cm}
\subfloat[(b)][]{
\includegraphics[trim=10mm 64mm 21mm 70mm,clip,scale=.4]{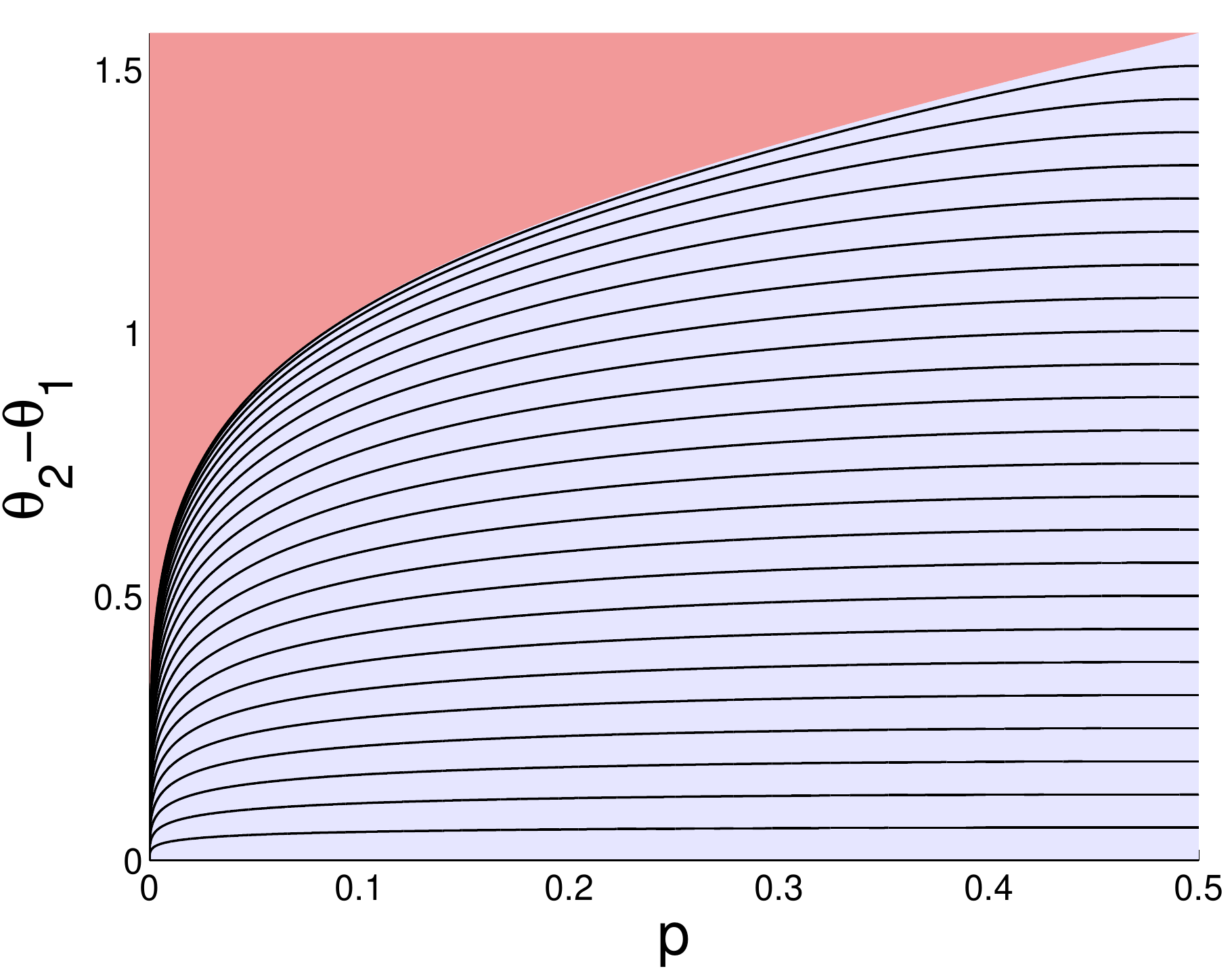}
}
\caption{(Color Online) Depiction of the ground state degeneracy of $H_n^\circ(\psi_+)$ where $\psi_+$ is of the form given in Eq.~\eqref{realstates}. The black lines are curves in the $(\theta_2-\theta_1,p)$ plane where $H_n^\circ(\psi)$ has ground state degeracy equal to $n+1$. Here we plot the curves for (a) $n=10$ and (b) $n=50$. For any point which does not lie on one of these curves (for $0<p\leq \frac{1}{2}$ and $\theta_2-\theta_1\in [0,\pi/2]$), the ground state degeneracy
of $H_n^\circ(\psi)$   is  two. We also show the gapped (red) and gapless (blue) regions for the chain with open boundary conditions. As $n\rightarrow\infty$ the black curves become dense in the blue region.}
\label{fig:Tncurves}
\end{figure}
{\em Gapped phase.}
In Sections~\ref{sec:mono},\ref{sec:gapped} 
we prove that the spectral gap of $H_n(\psi)$ is lower bounded by a
positive  constant 
independent of $n$
if the eigenvalues of $T_\psi$ have distinct magnitudes or if  both eigenvalues are equal to zero. 
Our starting point is a general method for bounding the spectral gap of frustration-free spin chains
due to Nachtergaele~\cite{Nachtergaele1996},
see Lemma~\ref{lemma:N} in Section~\ref{sec:gapped}.
To apply this method one has to manipulate expressions that involve the projector 
onto the ground space of $H_n(\psi)$ which we denote $G_n$.
The main technical difficulty that we had to overcome is a lack of an
explicit expression for $G_n$ which prevents us from straightforwardly applying Nachtergaele's bound.
Our proof is therefore  indirect and is based on establishing some features of the ground space which allow us to control $G_n$ sufficiently well. The key technical ingredient is  
a new operator inequality which expresses a monotonicity of the ground space projectors under the partial trace.
More precisely, we show that 
\begin{equation}
\label{Mintro}
\trace{}_n {(G_n)}\geq G_{n-1}.
\end{equation}
where the partial trace is taken over the $n$-th qubit. 
Using the fact that the Hamiltonians $H_n(\psi)$ are frustration-free one can easily
check that $\trace{}_n {(G_n)}$ and $G_{n-1}$ have the same support, that is,
Eq.~(\ref{Mintro}) is equivalent to saying that all non-zero eigenvalues of $\trace{}_n {(G_n)}$
are at least one. Our proof of this monotonicity property, presented in Section~\ref{sec:mono},
applies to general frustration-free chains  of qubits
composed of  rank-$1$ projectors. Neither translation-invariance nor the conditions 
of Theorem~\ref{thm:main} are needed for the proof of Eq.~(\ref{Mintro}). We note that Eq.~(\ref{Mintro}) differs from the well-known monotonicity property $G_n\le G_{n-1}\otimes I$. The latter 
follows trivially from the fact that $H_n(\psi)$ is frustration-free, whereas  Eq.~(\ref{Mintro}) holds for more subtle reasons. 

We proceed by showing that a quantum state which is completely mixed over the ground space of the $n$-qubit chain (i.e., proportional to the projector $G_n$)
exhibits an exponential decay of correlations for certain local observables, see
Lemma~\ref{lemma:decay} in Section~\ref{subs:corrfunc}. In Section~\ref{subs:regionexclusion} we use the decay of correlations and Eq.~(\ref{Mintro}) to prove several ``Region Exclusion" lemmas. Here we consider  a partition of the chain
into three or more regions and define local ground space projectors associated with each region.
Loosely speaking, the Region Exclusion lemmas state that the global ground space projector
associated with the entire chain can be approximated by a certain operator built from the
local  ground space projectors. The latter are defined on subsets of qubits where some of the chosen
regions are excluded from the chain (hence the name of the lemmas).
By repeatedly  applying the Region Exclusion lemmas
in Section~\ref{subs:finalgappedproof}
we arrive at the condition used in Nachtergaele's bound, thus proving a constant
lower bound on the gap.


\section{Structure of the ground space}
\label{sec:ground space}

In this section we describe the ground spaces of $H_n(\psi)$ and $H_n^\circ(\psi)$, the Hamiltonians for the chain with open and periodic boundary conditions respectively (defined in Eqs.~(\ref{H},\ref{Hcirc})). 

\subsection{Open boundary conditions}
\label{subs:open}
We first consider the Hamiltonian $H_n(\psi)$ for the chain with open boundary conditions. We begin with the simple case where $\psi=\psi_1\otimes \psi_2$ is a product state. It is always possible to choose 
the basis states $|0\rangle$ and $|1\rangle$ so that 
\[
|\psi\rangle =|1\rangle\otimes |v^\perp\rangle 
\]
where
\[
|v\rangle=c|0\rangle+s|1\rangle, \quad  |v^\perp\rangle = s^*|0\rangle - c^* |1\rangle, \quad
\mbox{and} \quad |c|^2+|s|^2=1.
\]
For each $i=1,\ldots,n$ define an
 $n$-qubit state $|g_i\rangle=|0^{i-1} v^\perp v^{n-i}\rangle$.
Also define $|g_0\rangle=|v^{\otimes n}\rangle$.
For example, choosing $n=4$ one gets
\[
\ba{rcrlllll}
|g_0\rangle &=&|&v&v&v&v&\rangle,\\
|g_1\rangle&=&|& v^\perp & v & v& v &\rangle,\\
|g_2\rangle&=&|& 0 & v^\perp & v & v &\rangle,\\
|g_3\rangle&=&| & 0 & 0 & v^\perp & v& \rangle,\\
|g_4\rangle&=&| &0 & 0 & 0 & v^\perp & \rangle.\\
\ea
\]
Loosely speaking, the states $g_i$ can be viewed as ``domain walls" 
where $|0\rangle$ and $|v\rangle$ represent 
two different values of a magnetization.
By direct inspection we see that $g_0,\ldots,g_n$ are pairwise orthogonal
ground states of $H_n(\psi)$. 
\begin{prop}
\label{prop:prod_basis}
Suppose $s\ne 0$. Then 
the states $g_0,\ldots,g_n$ form an orthonormal basis for the ground space of $H_n(\psi)$.
\end{prop}
\begin{proof}
It suffices to show that the ground space of $H_n(\psi)$
has dimension at most $n+1$. 
Define $\tilde{0} \equiv 0$ and $\tilde{1} \equiv v$.
Given any binary string $x=(x_1,\ldots,x_n)$, define 
$\tilde{x}\equiv (\tilde{x}_1,\ldots,\tilde{x}_n)$. Note that $|\tilde{0}\rangle$
and $|\tilde{1}\rangle$ are linearly independent since $s\ne 0$. Therefore
the  states $|\tilde{x}\rangle$, $x\in \{0,1\}^n$ 
form a basis (non-orthonormal) for 
the Hilbert space of $n$ qubits.
Suppose $|g\rangle$ is a ground state of $H_n(\psi)$. Then 
$|g\rangle=\sum_{x} a_x |\tilde{x}\rangle$
for some complex coefficients $a_x$.
A simple calculation shows that 
\[
{}_{i,i+1}\langle\psi|g\rangle = s^2 \sum_{x\, : \, (x_i,x_{i+1})=(1,0)}\;
a_x |\tilde{x}_1,\ldots,\tilde{x}_{i-1},\tilde{x}_{i+2},\ldots,\tilde{x}_n\rangle
\]
for any $i=1,\ldots,n-1$. On the other hand, ${}_{i,i+1}\langle\psi|g\rangle=0$
since $|g\rangle$ is a ground state of $H_n(\psi)$. This is possible only if 
$a_x=0$ for all strings $x$ that contain at least one consecutive pair $(1,0)$. 
Thus $|g\rangle$ belongs to a subspace spanned by vectors
$|0^iv^{n-i}\rangle$, where $i=0,\ldots,n$.
This shows that the ground subspace of $H_n(\psi)$ has dimension at most $n+1$.
\end{proof}

Now consider the case where $\psi$ is entangled. In this case we can still construct the ground space of $H_n(\psi)$ although, in contrast with the product state case, we are not able to obtain an orthonormal basis. The matrix $T_\psi$ defined in Eq.~(\ref{Tpsi})  plays a crucial role.  

One can easily check that $\det{(T_{\psi})}\neq0$ whenever $\psi$ is entangled
and
\begin{equation}
\label{Teps}
\langle \psi| (I\otimes T_\psi )= \det{(T_\psi)} \langle \epsilon|,
\end{equation}
 where $|\epsilon\rangle=|0,1\rangle-|1,0\rangle$ is the antisymmetric
state of two qubits. 
This shows that  the ground space of $H_2(\psi)=|\psi\rangle\langle \psi|_{1,2}$ is the image of the $2$-qubit symmetric subspace under the map $1\otimes T_\psi$. A similar characterization holds for $H_n(\psi)$ with $n>2$. In particular, define
\begin{equation}
\label{Talldef}
T^{\mathrm{all}}_\psi=I\otimes T_{\psi}\otimes T_{\psi}^2\otimes \ldots \otimes T_{\psi}^{n-1}.
\end{equation}
The following Proposition is a special case of a result presented in \cite{B06} (and has been used previously in, e.g., \cite{unfrustratedarea}).
\begin{prop}
\label{gsdegen}
Suppose $\det{(T_{\psi})}\neq0$. Then the ground space of $H_n(\psi)$ is the image of the $n$-qubit symmetric subspace under the linear map $T^{\mathrm{all}}_\psi$.
\end{prop}
\begin{proof}
Using Eq.~\eqref{Teps} and the fact that  $M\otimes M|\epsilon\rangle=\det{(M)}|\epsilon\rangle$ we get
\begin{equation}
(T_{\psi}^{\mathrm{all}})^\dagger|\psi\rangle\langle\psi|_{j,j+1}T_{\psi}^{\mathrm{all}}=|\epsilon\rangle\langle\epsilon|_{j,j+1}\otimes B_{j}\label{eq:T_all_eqn}
\end{equation}
for each $j=1,\ldots,n-1$, where $B_{j}$ is a positive operator
acting on qubits in the set $[n]\setminus\{j,j+1\}$. From Eq.~\eqref{eq:T_all_eqn}
we see that the nullspace of $(T_{\psi}^{\mathrm{all}})^\dagger H_n(\psi)T_{\psi}^{\mathrm{all}}$
is equal to the symmetric subspace. The result follows since $T_{\psi}^{\mathrm{all}}$ is invertible.
\end{proof}
Combining Propositions~\ref{prop:prod_basis},\ref{gsdegen} 
and noting that the symmetric subspace of $n$ qubits has dimension $n+1$,
we conclude that the ground space of $H_n(\psi)$ has dimension $n+1$
for almost any choice of $\psi$ (the only exception is when $s=0$ and $\psi$ is a symmetric product state).

\subsection{Periodic boundary conditions}
\label{subs:periodic}

We now consider the Hamiltonian $H_n^\circ(\psi)$ for the chain with periodic boundary conditions.
It is well-known  that  $H^{\circ}_n(\psi)$ is frustration-free for any choice of $\psi$,
see for instance~\cite{B06,laumann2010phase,BMR2010}. However in this paper we will only need to deal with periodic boundary conditions  in the 
case where $\psi$ is an entangled state. Accordingly, in this section we 
assume that  $\det{(T_{\psi})} \neq 0$. For any such $\psi$ we compute the dimension of the zero energy ground space
of  $H^{\circ}_n(\psi)$. We will see that it takes different values depending on the choice of $\psi$. This contrasts with the open boundary chain which has ground space dimension $n+1$ whenever $\psi$ is entangled.

Here and throughout the paper we use the symbol $\sim$ to mean proportional to.
\begin{prop}
\label{gs_counting} 
Suppose $T_{\psi}^{n}\sim I$. Then 
the   ground space of  $H^{\circ}_n(\psi)$ has dimension $n+1$.
Otherwise,  $H^{\circ}_n(\psi)$ has a two-fold degenerate
ground space. 
\end{prop}
\begin{proof}
Note that $H^{\circ}_n(\psi)$ has the same rank as
\begin{equation}
T_{\psi}^{\mathrm{all}\dagger}H^{\circ}_n(\psi)T_{\psi}^{\mathrm{all}}=T_{\psi}^{\mathrm{all}\dagger}H_n(\psi)T_{\psi}^{\mathrm{all}}+T_{\psi}^{\mathrm{all}\dagger}|\psi\rangle\langle\psi|_{n,1}T_{\psi}^{\mathrm{all}}.\label{eq:Q_periodic}
\end{equation}
where $T^{\mathrm{all}}_\psi$ is given by Eq.~\eqref{Talldef}.  Both terms on the right-hand side are positive semidefinite and, by Proposition \ref{gsdegen}, the nullspace of the first term is the symmetric subspace. If $T_{\psi}^{n}\sim I$
then the second term in Eq.~\eqref{eq:Q_periodic} can be written
as $|\epsilon\rangle\langle\epsilon|_{n,1}\otimes B_{n}$ where $B_{n}$
is positive and $|\epsilon\rangle=|0,1\rangle-|1,0\rangle$. Since this term annihilates every state in the
symmetric subspace we see that in this case the nullspace of Eq.~\eqref{eq:Q_periodic} is $(n+1)$-dimensional. 

If $T_{\psi}^{n}$ is not proportional to the identity we show that
there are exactly two states in the symmetric subspace which are annihilated
by the second term in Eq.~\eqref{eq:Q_periodic}. We consider
two cases depending on whether or not $T_{\psi}^{n}$ is defective (has only one eigenvector).

First consider the case where $T_{\psi}^{n}$ has two linearly independent
eigenvectors $|v_{1}\rangle,|v_{2}\rangle$. 
Note that the last term in Eq.~(\ref{eq:Q_periodic}) projects
qubits $n,1$ onto 
a state 
\[
|\phi\rangle=(T_{\psi}^{n-1\dagger}\otimes I )|\psi\rangle
\sim  (T_{\psi}^{n\dagger}\otimes I)|\epsilon\rangle=\left(T_{\psi}^{n\dagger}|0\rangle\right)|1\rangle-\left(T_{\psi}^{n\dagger}|1\rangle\right)|0\rangle.
\]
The last equality  makes it clear that $|\phi\rangle$
and $|\epsilon\rangle$ are linearly independent whenever $T_{\psi}^{n}$
is not proportional to the identity. Thus
the nullspace of Eq.~(\ref{eq:Q_periodic}) is spanned by $n$-qubit
symmetric states that are orthogonal to 
$|\phi\rangle$ on any pair of qubits.
One can easily check that 
the only two-qubit symmetric states orthogonal to  $|\phi\rangle$
are $|v_1 \otimes v_1\rangle$ and $|v_2\otimes v_2\rangle$.
Likewise, one can  check that the only 
$n$-qubit symmetric states orthogonal to $|\phi\rangle$ on any pair of qubits 
are linear combinations of $|v_1\rangle^{\otimes n}$ and $|v_2\rangle^{\otimes n}$.
Thus Eq.~(\ref{eq:Q_periodic})  has a two-dimensional nullspace
and therefore the same is true for $H_n^\circ(\psi)$.

Next suppose $T_{\psi}^{n}$ is defective, i.e., has only one eigenvector.
Let us work in a basis where $|0\rangle$ is this eigenvector, so
\[
T_{\psi}^{n}=\left(\begin{array}{cc}
b & a\\
0 & b
\end{array}\right)
\]
 for some $a,b\in\mathbb{C}$ with $a\neq0$. Then 
 the last term in Eq.~(\ref{eq:Q_periodic}) projects onto a state
\[
|\phi\rangle=(T_{\psi}^{n-1\dagger}\otimes I )|\psi\rangle \sim
(T_{\psi}^{n\dagger}\otimes I )|\epsilon\rangle=b^{\ast}|\epsilon\rangle+a^{*}|11\rangle.
\]
Since $a\neq 0$, the states $|\phi\rangle$ and $|\epsilon\rangle$ span
the same subspace as $|11\rangle$ and $|\epsilon\rangle$.
Therefore the nullspace of Eq.~(\ref{eq:Q_periodic}) is spanned by $n$-qubit
symmetric states that are orthogonal to $|11\rangle$
on any pair of qubits. One can easily check that the only such states
are linear combinations of 
$|0\rangle^{\otimes n}$ and the $n$-qubit W-state
\[
|100\ldots0\rangle+|010\ldots0\rangle+\ldots+|00\ldots01\rangle.
\]
Thus Eq.~(\ref{eq:Q_periodic})  has a two-dimensional nullspace
and therefore the same is true for $H_n^\circ(\psi)$.
\end{proof}


\section{Gapless phase}
\label{sec:gapless}

In this section we prove the first part of 
Theorem~\ref{thm:main}, namely, 
\begin{gaplessphasethm}
\label{gaplessthm}
Suppose the eigenvalues of $T_{\psi}$  have the same non-zero absolute value.
Then   $\gamma(\psi,n)\leq 1/(n-1)$ for all $n\ge 2$.
\end{gaplessphasethm}
\noindent
Recall that $\gamma(\psi,n)$ denotes  the smallest non-zero eigenvalue of $H_n(\psi)$. In addition we
write $\gamma^{\circ}(\psi,n)$  for the smallest non-zero eigenvalue of the Hamiltonian $H^{\circ}_n(\psi)$ with periodic boundary conditions, see Eq.~(\ref{Hcirc}).

To prove the gapless phase theorem we use the following lemma, proven by Knabe \cite{K88}, which relates the smallest
non-zero eigenvalues of the chains with  periodic and open boundary conditions. Knabe's result, presented in Section 2 of reference \cite{K88}, applies to more general frustration-free spin chains but here we specialize to the case at hand.
\begin{Lemma}
[Knabe \cite{K88}]For all $m\geq n>2$, 
\begin{equation}
\label{knabe}
\gamma^{\circ}(\psi,m)\geq\frac{n-1}{n-2}\left(\gamma(\psi,n)-\frac{1}{n-1}\right).
\end{equation}
\label{lemma:knabe}
\end{Lemma}
This lemma was originally proposed as a technique for proving that
the periodic chain  is gapped in the
thermodynamic limit. This follows from the lemma if 
one can show that 
there exists a finite $n$ for which the open chain  has a gap strictly larger than $1/(n-1)$. Here
we apply the lemma in the opposite direction. We use the following strategy which works for some (but not all) $\psi$ satisfying the conditions of the gapless phase theorem. First we  apply the argument sketched in  Section~\ref{sketch} to show that 
 $\gamma^{\circ}(\psi,m)$  can take arbitrarily small values for large enough $m$. Then we apply Knabe's lemma
to infer that $\gamma(\psi,n)\le 1/(n-1)$ for any $n>2$ since otherwise Eq.~(\ref{knabe}) would provide a constant lower bound on $\gamma^{\circ}(\psi,m)$
for all $m\ge n$, leading to a contradiction. Note also that $\gamma(\psi,2)=1$ since $H_2(\psi)=|\psi\rangle\langle \psi|$.

For some states $\psi$ we are not able to use the above strategy directly; however in these cases we choose a state $\phi$ which can be taken arbitrarily close to $\psi$ for which the strategy can be applied. The result for $\psi$ then follows by continuity. In order to handle these cases (and for other portions of the proof) we will need the following straightforward bound on how much the eigenvalues of $H_n(\psi)$ (or $H^{\circ}_n(\psi))$ can change as $\psi$ varies. Write
\[
e_1(\psi,n)\leq e_2(\psi,n)\leq \ldots \leq e_{2^n}(\psi,n) \quad \text{and} \quad e^{\circ}_1(\psi,n)\leq e^{\circ}_2(\psi,n)\leq \ldots \leq e^{\circ}_{2^n}(\psi,n) 
\]
for the eigenvalues of $H_n(\psi)$ and $H_n^\circ (\psi)$ respectively. 

\begin{prop}
\label{Weyl}Let $\psi$ and $\phi$ satisfy $\|\psi\|=\|\phi\|=1$.
Then
\[
\left|e_{j}(\psi,n)-e_{j}(\phi,n)\right|\leq2n\left\Vert \psi-\phi \right\Vert \quad\text{and}\quad\left|e_{j}^{\circ}(\psi,n)-e_{j}^{\circ}(\phi,n)\right|\leq2n \left\Vert \psi-\phi \right\Vert
\]
 for each $j=1,\ldots,2^{n}.$\end{prop}
\begin{proof}
The proof of the two inequalities is almost identical so here we prove
only the first one. We use the Weyl inequality for perturbed eigenvalues
(see for example Corollary III.2.6 of reference \cite{BH97}) which
in this case says $\left|e_{j}(\psi,n)-e_{j}(\phi,n)\right|\leq\left\Vert H_n(\psi)-H_n(\phi)\right\Vert .$
To complete the proof we bound 
\begin{align*}
\left\Vert H_n(\psi)-H_n(\phi)\right\Vert  & \leq\sum_{i=1}^{n-1}\left\Vert |\psi\rangle\langle\psi|_{i,i+1}-|\phi\rangle\langle\phi|_{i,i+1}\right\Vert \\
 & =\frac{(n-1)}{2}\left\Vert \left(|\psi\rangle-|\phi\rangle\right)\left(\langle\psi|+\langle\phi|\right)+\left(|\psi\rangle+|\phi\rangle\right)\left(\langle\psi|-\langle\phi|\right)\right\Vert \\
 & \leq(n-1)\left\Vert \left(|\psi\rangle-|\phi\rangle\right)\left(\langle\psi|+\langle\phi|\right)\right\Vert \\
 & \leq2(n-1)\left\Vert |\psi\rangle-|\phi\rangle\right\Vert 
\end{align*}
where in the last line we used the fact that $\left\Vert |\psi\rangle+|\phi\rangle\right\Vert \leq2$
(since $|\psi\rangle$ and $|\phi\rangle$ are normalized). 
\end{proof}

We now proceed to the proof of the gapless phase theorem.
\begin{proof}
First we claim that the eigenvalues of $H_n(\psi)$ and 
absolute values  of the eigenvalues of 
$T_\psi$ are invariant
under a transformation $\psi \to (U\otimes U)\psi$ 
where $U$ is an arbitrary single-qubit unitary operator.  
 Indeed, let $\psi'=(U\otimes U)\psi$.
Then Eq.~(\ref{H}) implies  
$H_n(\psi')=U^{\otimes n} H_n(\psi) (U^\dag)^{\otimes n}$
and Eq.~(\ref{Teps}) implies $T_{\psi'}=(\det{U})^{-1} \cdot U T_\psi U^\dag$.
Thus the eigenvalues of $H_n(\psi')$ and magnitudes of the eigenvalues of $T_{\psi'}$ do not depend on $U$.
We shall use the freedom in choosing $U$ to bring $\psi$ into a certain canonical form
as defined below.
\begin{prop}
\label{prop:canonical}
For any $|\psi\rangle \in \CC^2\otimes \CC^2$ there exists 
a single-qubit unitary $U$ 
such that 
\begin{equation}
\label{psi_canonical}
(U\otimes U)|\psi\rangle=(\alpha+i\beta)|0,1\rangle+(\alpha+i\gamma)|1,0\rangle+\delta|1,1\rangle
\end{equation}
for some real coefficients $\alpha,\beta,\gamma,\delta$.
\end{prop}
Since the proof is rather straightforward, we shall postpone it until the end of this section.
From now on we can assume that $\psi$ has the canonical form as in the
right-hand side of Eq.~(\ref{psi_canonical}). Substituting this canonical form into Eq.~(\ref{Tpsi}) one gets
\begin{equation}
\label{transfer}
T_{\psi}=\left(\begin{array}{cc}
\alpha-i\beta & \delta\\
0 & -(\alpha-i\gamma)
\end{array}\right).
\end{equation}
The eigenvalues of this matrix are $\alpha-i\beta$ and $-(\alpha-i\gamma)$, with magnitudes $\sqrt{\alpha^2+\beta^2}$ and
$\sqrt{\alpha^2+\gamma^2}$. By assumption of the theorem,
the eigenvalues have the same magnitude and thus $\gamma=\pm\beta$.
We consider the two cases $\gamma=\pm\beta$ separately and we show that $\gamma(\psi,n)\leq \frac{1}{n-1}$ in each case.

\noindent \begin{flushleft}
\textbf{Case 1: $\gamma=\beta$}
\par\end{flushleft}

Fix $n$ and let $m\geq n$ be even. Setting $\beta=\gamma$ in Eq.~\eqref{transfer} and taking the square we see that $T_\psi^2\sim I$. Since $m$ is even we get $T_{\psi}^{m}\sim I$ and therefore
$e_{3}^{\circ}(\psi,m)=0$ by Proposition \ref{gs_counting}.
Let $|\phi_{m}\rangle$ be a normalized state which satisfies 
\[
\left\Vert \phi_{m}-\psi\right\Vert \leq\frac{1}{m^{2}}
\]
 and such that the eigenvalues of $T_{\phi_m}$ have different magnitudes and are both non-zero. This guarantees that $T_{\phi_{m}}^{m}$ is not proportional to the identity and $\det(T_{\phi_m})\neq 0$. Then, by Propositions \ref{gs_counting} and \ref{Weyl}, 
\[
\gamma^{\circ}(\phi_{m},m)=e_{3}^{\circ}(\phi_{m},m)=e_{3}^{\circ}(\phi_{m},m)-e_{3}^{\circ}(\psi,m)\leq\frac{2}{m}.
\]
Applying Lemma \ref{lemma:knabe} gives
\[
e_{n+2}(\phi_{m},n)=\gamma(\phi_{m},n)\leq\frac{1}{n-1}+\frac{n-2}{n-1}\gamma^{\circ}(\phi_{m},m)\leq\frac{1}{n-1}+\frac{2}{m}\left(\frac{n-2}{n-1}\right)
\]
for all even $m\geq n$, and using Propositions \ref{gs_counting} and \ref{Weyl}
again we have 
\begin{align*}
\gamma(\psi,n) & =e_{n+2}(\phi_{m},n)+\left(e_{n+2}(\psi,n)-e_{n+2}(\phi_{m},n)\right)\\
 & \leq e_{n+2}(\phi_{m},n)+\frac{2n}{m^{2}}\\
 & \leq\frac{1}{n-1}+\frac{2}{m}\left(\frac{n-2}{n-1}\right)+\frac{2n}{m^{2}}.
\end{align*}
The result follows since $m\geq n$ can be taken arbitrarily large.

\noindent \begin{flushleft}
\textbf{Case 2: $\gamma=-\beta$}
\par\end{flushleft}

Let $\alpha+i\beta=re^{-i\pi\left(\theta+\frac{1}{2}\right)}$ where $\theta$ and $r$ are positive real numbers. Here $r<1$ since $\psi$ is normalized and $r>0$ since the eigenvalues of $T_\psi$  are assumed to be non-zero. We first consider the case where
$\theta$ is irrational. In this case the convergents of the continued
fraction expansion of $\theta$ give sequences of positive integers $\{p_{j}\}$
and $\{q_{j}\}$ with 
\begin{equation}
\left|\frac{p_{j}}{q_{j}}-\theta\right|\leq\frac{1}{q_{j}^{2}}\label{eq:continued_fracs},
\end{equation}
 $\mathrm{gcd}(p_j,q_j)=1$, and where $\{q_{j}\}$ diverges. Here we shall omit the first two convergents obtained by the standard continued fraction expansion, in order to guarantee that the sequence $\{q_j\}$ is strictly increasing and $q_j\geq 2$ for all $j$. 

 Define $\theta_{j}=\frac{p_{j}}{q_{j}}$, and let 
\[
|\Psi_{j}\rangle=re^{-i\pi\left(\theta_{j}+\frac{\pi}{2}\right)}|0,1\rangle+re^{i\pi\left(\theta_{j}+\frac{\pi}{2}\right)}|1,0\rangle+\delta|1,1\rangle.
\]
 Then 
\begin{equation}
\left\Vert \Psi_{j}-\psi\right\Vert  =r\sqrt{2}\left|e^{i\pi\left(\theta_{j}-\theta\right)}-1\right| \leq r\sqrt{2}\left|\pi\left(\theta_{j}-\theta\right)\right|  \leq\frac{\sqrt{2}\pi}{q_{j}^{2}}\label{eq:theta_1}
\end{equation}
 where we used the inequality $\left|e^{ix}-1\right|\leq\left|x\right|$,
Eq.~\eqref{eq:continued_fracs}, and the fact that $r<1$. Note
that 
\[
T_{\Psi_{j}}=\left(\begin{array}{cc}
ire^{i\pi\theta_{j}} & \delta \\
0 & ire^{-i\pi\theta_{j}}
\end{array}\right)
\]
has eigenvalues $E_1=ire^{i\pi\theta_{j}}$ and  $E_2=ire^{-i\pi\theta_{j}}$.  We have $E_1\neq E_2$, which follows from the fact that $\theta_j$ is not an integer, since $q_j\geq 2$. Thus $T_{\Psi_j}$ is diagonalizable. Furthermore, $E_1^{q_j}=E_2^{q_j}$, and therefore (using the fact that it is diagonalizable) $T_{\Psi_{j}}^{q_{j}}\sim I$. Hence $e_{3}^{\circ}(\Psi_{j},q_{j})=0$
by Proposition \ref{gs_counting}. On the other hand $T_{\psi}^{q_{j}}$
is not proportional to the identity since $\theta$ is irrational,
hence $\gamma^{\circ}(\psi,q_{j})=e_{3}^{\circ}(\psi,q_{j})$.
Using these facts and Proposition \ref{Weyl} we have 
\begin{align*}
\gamma^{\circ}(\psi,q_{j}) & =\left(e_{3}^{\circ}(\psi,q_{j})-e_{3}^{\circ}(\Psi_{j},q_{j})\right)\leq 2q_{j}\left\Vert \Psi_{j}-\psi\right\Vert \leq\frac{2\sqrt{2}\pi}{q_{j}}.
\end{align*}
 Now for all $j$ such that $q_{j}\geq n$ we get 
\[
\gamma(\psi,n)\leq\frac{1}{n-1}+\frac{n-2}{n-1}\gamma^{\circ}(\psi,q_{j})\leq\frac{1}{n-1}+\left(\frac{n-2}{n-1}\right)\frac{2\sqrt{2}\pi}{q_{j}}
\]
 and hence $\gamma(\psi,n)\leq\frac{1}{n-1}$ since the sequence $\{q_{j}\}$
diverges and the second term can be made arbitrarily small.

It remains to consider the case where $\theta$ is rational. In this
case, for any $\epsilon$ we may choose $\theta'$ to be an irrational
number satisfying $\left|\theta'-\theta\right|\leq\epsilon$. Letting
\[
|\phi\rangle=re^{-i\pi\left(\theta'+\frac{\pi}{2}\right)}|0,1\rangle+re^{i\pi\left(\theta'+\frac{\pi}{2}\right)}|0,1\rangle+\delta|1,1\rangle
\]
we may now apply the above proof to get $\gamma(\phi,n)=e_{n+2}(\phi,n)\leq\frac{1}{n-1}$.
Now using Proposition \ref{Weyl} we get 
\begin{align*}
\gamma(\psi,n) & =e_{n+2}(\phi,n)+\left(e_{n+2}(\psi,n)-e_{n+2}(\phi,n)\right)\leq\frac{1}{n-1}+2n\left\Vert |\phi\rangle-|\psi\rangle\right\Vert \\
 & \leq\frac{1}{n-1}+2nr\sqrt{2}\pi\epsilon
\end{align*}
where in the second line we used the same reasoning as Eq.~\eqref{eq:theta_1}. Since $\epsilon$ can be chosen arbitrarily
small we get $\gamma(\psi,n)\leq\frac{1}{n-1}$.
\end{proof}
Finally, let us prove Proposition~\ref{prop:canonical}.
\begin{proof}
Recall that a transformation $\psi\to (U\otimes U)\psi$
maps $T_\psi$ to 
$(\det{U})^{-1} \cdot U T_\psi U^\dag$. Here $U$ is an arbitrary unitary operator. 
We shall choose a sequence of such transformations that bring $T_\psi$
into the canonical form defined in Eq.~(\ref{transfer}) which is equivalent
to Eq.~(\ref{psi_canonical}). First, choose $U$ such that
$|0\rangle$ is an  eigenvector  of $T_\psi$.
This is always possible since any complex matrix has at least one eigenvector. 
Now we can assume that 
\[
T_{\psi}=\left[ \begin{array}{cc}
\mu_1 & \delta \\
0 & \mu_2 \\
\end{array}\right]
\]
for some complex coefficients $\mu_1,\mu_2,\delta$.
Next choose 
$U=e^{-i\theta/2}I$ where the phase $\theta$ 
satisfies $\mathrm{Re}(e^{i\theta} (\mu_1+\mu_2))=0$.
This maps $T_\psi$ to $e^{i\theta}T_\psi$ and now we can assume that 
$\mathrm{Re}(\mu_1+\mu_2)=0$.
Finally, choosing
$U=\mathrm{diag}(e^{i\theta},e^{-i\theta})$ one can 
map $\delta$ to $e^{2i\theta} \delta$ without changing $\mu_1,\mu_2$.
Thus we can make $\delta$ real. This brings $T_\psi$
into the canonical form defined in Eq.~(\ref{transfer}).
\end{proof}


\section{Monotonicity under the partial trace}
\label{sec:mono}

In this section we establish a relationship between
the ground space projectors describing a chain of $n$  and $n-1$ qubits.
Our result holds in a more general setting than considered elsewhere in this paper since we do not assume translation invariance.

Let  $\psi_1,\psi_2,\ldots,\psi_{m-1}$ be an arbitrary sequence
of normalized two-qubit states. For each $n=2,\ldots,m$
define a Hamiltonian
\begin{equation}
\label{Hgeneral}
H_n (\psi_1,\psi_2,\ldots, \psi_{n-1})=\sum_{j=1}^{n-1} |\psi_j\rangle\langle \psi_j|_{j,j+1}
\end{equation}
which describes a chain of $n$ qubits. This Hamiltonian is frustration-free for any choice of  $\psi_1,\ldots,\psi_{n-1}$ (this follows directly from Proposition \ref{dimincreasing} given below). Let $\calG_n$ be the ground space of $H_n (\psi_1,\psi_2,\ldots, \psi_{n-1})$
and $G_n$ be the projector onto $\calG_n$.
We adopt the convention $H_1=0$ and $\calG_1=\CC^2$.

First, we note that
$\calG_n \subseteq \calG_{n-1}\otimes \CC^2$
due to
the fact that the considered Hamiltonians are frustration-free.
This results in a trivial
monotonicity property  $G_n\le G_{n-1}\otimes I$. 
Below we prove that one also has a different type of
monotonicity, namely $\trace{}_n(G_n)\ge G_{n-1}$,
where $\trace{}_n$ represents the partial trace over the $n$-th qubit.
\begin{Lemma}[\bf Monotonicity]
For each $n=2,\ldots,m$ one has
\label{lemma:mono}
\begin{equation}
\label{mono}
\trace_n(G_n) \ge G_{n-1}.
\end{equation}
\end{Lemma}
\noindent
\noindent
Given the simplicity and generality of Eq.~(\ref{mono}),
one may be tempted to ask whether it holds for some trivial reason
unrelated to the structure of the considered Hamiltonians.
We have observed numerically that  Eq.~(\ref{mono})
can be false if $G_{n-1}$ and $G_n$ are chosen as projectors onto random  linear subspaces
$\calG_{n-1}\subseteq (\CC^2)^{\otimes (n-1)}$ and $\calG_n\subseteq \calG_{n-1}\otimes \CC^2$,
even if  the dimensions of $\calG_{n-1}$ and $\calG_n$ match those of the ground subspaces of $H_{n-1}(\psi_1,\ldots,\psi_{n-2})$ and $H_n(\psi_1,\ldots,\psi_{n-1})$. Thus any proof of the monotonicity property must exploit
 the special structure 
of the projectors $G_n$. In the absence of  an explicit formula for $G_n$, 
one has to rely on some indirect arguments in order to derive Eq.~(\ref{mono}).
This partially explains why the proof of the lemma given below is rather cumbersome. 

\begin{proof}[\bf Proof of Lemma~\ref{lemma:mono}]
We use induction in $n$. The base of the induction is $n=2$.
In this case $G_2=I-|\psi_1\rangle\langle \psi_1|$ and thus
$\trace_2(G_2)=2I-\trace_2( |\psi_1\rangle\langle \psi_1|) \ge I=G_1$.
Here we used the fact that the partial trace of any two-qubit state is a 
density matrix which has eigenvalues at most one. 
We now prove the induction step.  For brevity denote $\psi\equiv \psi_{n-1}$ such that the last
term in $H_n(\psi_1,\ldots,\psi_{n-1})$ is $|\psi\rangle\langle\psi|_{n-1,n}$.

First consider the case where $\psi$ is unentangled, that is, $\psi=\alpha\otimes\beta$ for some single-qubit states $\alpha,\beta$. In this case the result follows trivially without using the inductive hypothesis, since $\calG_{n-1}\otimes\beta^\perp\subseteq \calG_n$ which implies $G_n\geq G_{n-1}\otimes |\beta^\perp\rangle\langle\beta^\perp|$ and thus $\trace{(G_n)}\geq G_{n-1}$.

In the remainder of the proof we consider the case where $\psi$ is entangled (i.e., not a product state). Write the Schmidt decomposition of $\psi$ as
\begin{equation}
\label{psiSchmidt}
|\psi\rangle=\sqrt{p_0}|w_0\rangle|v_0\rangle+\sqrt{p_1}|w_1\rangle|v_1\rangle
\end{equation}
where $\langle w_i|w_j\rangle=\langle v_i|v_j\rangle=\delta_{ij}$ and $p_0,p_1 >0$ with $p_0+p_1=1$.

Let $G_n^\perp=I-G_n$. Obviously, $\trace_n (G_n)=2I-\trace_n(G_n^\perp)$.
Furthermore, the trivial monotonicity $\calG_n \subseteq \calG_{n-1}\otimes \CC^2$
implies that 
$\trace_n(G_n)$ has all of its support on $\calG_{n-1}$,
that is, $\trace_n(G_n) =\trace_n(G_n)G_{n-1}=G_{n-1}\trace_n(G_n)$.
Define an operator
\begin{equation}
\label{R}
R_n\equiv G_{n-1} \trace_n(G_n^\perp) G_{n-1}.
\end{equation}
The above implies that 
\begin{equation}
\label{R1}
\trace_n(G_n)=G_{n-1} \trace_n(G_n) G_{n-1} = 2G_{n-1} - R_n \ge (2-\|R_n\|) G_{n-1}.
\end{equation}
Thus it suffices to prove that $\|R_n\|\le 1$.

Choose an arbitrary orthonormal basis 
\begin{equation}
\label{basis1}
g_1,g_2,\ldots,g_{r} \in \calG_{n-2}, \quad \quad \langle g_\alpha|g_\beta\rangle=\delta_{\alpha,\beta}.
\end{equation}
Also choose an arbitrary orthonormal basis 
\begin{equation}
\label{basis2}
h_1,h_2,\ldots,h_s \in \calG_{n-1}, \quad \quad \langle h_i|h_j\rangle=\delta_{i,j}.
\end{equation}
In general the dimensions $r$ and $s$ of the spaces $\calG_{n-2}$ and $\calG_{n-1}$ will depend on the states $\psi_1,\ldots,\psi_{n-2}$ but we will not need an explicit expression for them. We will however need to use the fact that $s>r$, which we now establish. The following proposition is a special case of the result presented in \cite{unfrustrated}; we include a proof here for completeness.
\begin{prop}
\label{dimincreasing}
Let $D_n$ be the dimension of $\calG_{n}$. Then $D_n>D_{n-1}$ for all $2\leq n\leq m$.
\end{prop}
\begin{proof}
 Recall our convention that $\calG_1=\CC^2$, so $D_1=2$. On the other hand $H_2(\psi_1)=|\psi_1\rangle\langle \psi_1|_{1,2}$ and $D_2=3$, which confirms $D_2>D_1$. We now establish that $D_n-D_{n-1}\geq D_{n-1}-D_{n-2}$ for all $n\geq 3$. This is sufficient to complete the proof since it implies $D_n-D_{n-1}\geq D_2-D_1=1$. 

Let $\phi$ be a general state in $\calG_n$, with $n\geq 3$. Let $\gamma_1,\ldots \gamma_{D_{n-1}}$ be an orthonormal basis for $\calG_{n-1}$ and let $\kappa_1,\ldots,\kappa_{D_{n-2}}$ be an orthonormal basis for $\calG_{n-2}$. We can write
\[
|\phi\rangle=\sum_{i=1}^{D_{n-1}} f_{i,0}|\gamma_i\rangle|0\rangle+ f_{i,1}|\gamma_i\rangle|1\rangle
\]
for some complex coefficients $\{f_{i,z}\}$. The fact that the Hamiltonian $H_n(\psi_1,\ldots,\psi_{n-1})$ is frustration-free implies
\[
\calG_n=\left(\calG_{n-1}\otimes \CC^2\right)\cap\left(\calG_{n-2}\otimes\psi_{n-1}^\perp\right)
\]
and thus the dimension of $\calG_n$ is the number of linearly independent solutions to the equations
\begin{equation}
\langle\kappa_j \otimes \psi_{n-1} |\phi\rangle=0 \quad \text{for all}  \quad j=1,\ldots,D_{n-2}.
\end{equation}
This is a set of $D_{n-2}$ linear equations for the $2D_{n-1}$ variables $\{f_{i,z}\}$. The number of linearly independent solutions satisfies $D_n\geq 2D_{n-1}-D_{n-2}$, or equivalently $D_n-D_{n-1}\geq D_{n-1}-D_{n-2}$.
\end{proof}

Define $r\times s$  matrices 
\begin{equation}
\label{gram}
(M_0)_{\alpha,i}=\langle g_\alpha \otimes w_0|h_i\rangle \quad \mbox{and} \quad
(M_1)_{\alpha,i}=\langle g_\alpha \otimes w_1|h_i\rangle.
\end{equation}
where $w_0,w_1$ are defined in Eq.~(\ref{psiSchmidt}). The trivial monotonicity $\calG_{n-1}\subseteq \calG_{n-2}\otimes \CC^2$ implies
\begin{equation}
\label{norm1}
M_0^\dag M_0 + M_1^\dag M_1 = I_s.
\end{equation}
Here and below $I_q$ denotes the identity matrix of dimension $q$. 
Furthermore, expressing $G_{n-1}=\sum_{i=1}^s |h_i\rangle\langle h_i|$ 
one gets 
\[
\langle g_\alpha| \trace_{n-1}(G_{n-1}) |g_\beta\rangle
= \langle g_\alpha \otimes w_0 | G_{n-1} |g_\beta \otimes w_0\rangle + 
\langle g_\alpha \otimes w_1 | G_{n-1} |g_\beta \otimes w_1\rangle
=(M_0 M_0^\dag + M_1 M_1^\dag)_{\alpha,\beta}.
\]
Since $\trace_{n-1}(G_{n-1})\ge G_{n-2}$ by  the induction hypothesis,
we infer that 
\begin{equation}
\label{norm2}
M_0 M_0^\dag + M_1 M_1^\dag \ge I_{r}.
\end{equation}
The usefulness of the matrices $M_0,M_1$ comes from the following facts. Recall that we defined $R_n=G_{n-1} \trace{}_n(G_n^\perp)G_{n-1}$.
\begin{prop}
\label{prop:R}
Suppose $\psi$ is entangled. Then the matrix of the operator $R_n$ in the chosen basis $\{h_1,\ldots,h_s\}$ of $\calG_{n-1}$ can be written as
\begin{equation}
\label{Rsimple1}
R_n=p_0 M_0^\dag ( p_0 M_0 M_0^\dag + p_1 M_1 M_1^\dag )^{-1} M_0
+ p_1 M_1^\dag ( p_0 M_0 M_0^\dag + p_1 M_1 M_1^\dag )^{-1} M_1,
\end{equation}
where $p_0,p_1>0$ are defined by Eq.~(\ref{psiSchmidt}).
\end{prop}

\begin{prop}
\label{prop:tricky}
Let $r,s$ be arbitrary positive integers with $s\geq r$. Let $M_0,M_1$ be arbitrary matrices of size $r\times s$ satisfying $M_0^\dag M_0 + M_1^\dag M_1=I_s$ and $M_0M_0^\dag + M_1 M_1^\dag \ge I_r$. Let $p_0,p_1> 0$ be any real positive numbers. Then the operator $R_n$ defined by Eq.~(\ref{Rsimple1}) satisfies
$\|R_n\|\le 1$.
\end{prop}
Combining Propositions~\ref{prop:R},\ref{prop:tricky}, the inequality $s>r$ proved in Proposition~\ref{dimincreasing}, and Eqs.~(\ref{norm1},\ref{norm2})
one gets  $\|R_n\| \le 1$. The lemma follows from   Eq.~(\ref{R1}).
\end{proof}
In the rest of this section we prove the above propositions. 
\begin{proof}[\bf Proof of Proposition~\ref{prop:tricky}]
 Denoting $x=p_1/p_0$ one can rewrite $R_n$ as
\begin{equation}
\label{eq0}
R_n=M_0^\dag ( M_0 M_0^\dag + x M_1 M_1^\dag)^{-1}  M_0 + x M_1^\dag ( M_0 M_0^\dag + x M_1 M_1^\dag)^{-1}  M_1.
\end{equation}
By symmetry we can assume that $x\ge 1$. Then Eq.~(\ref{norm2}) implies
\[
M_0M_0^\dag + xM_1 M_1^\dag = M_0M_0^\dag + M_1M_1^\dag + (x-1) M_1 M_1^\dag \ge I_{r} + 
(x-1)M_1M_1^\dag.
\]
Since the function $f(y)=-1/y$ is operator monotone, we arrive at
$R_n\le S_0+S_1$, where
\begin{equation}
\label{eq1}
S_0=M_0^\dag(I_{r} + (x-1) M_1 M_1^\dag)^{-1} M_0 \quad \mbox{and} \quad
S_1=xM_1^\dag (I_{r} + (x-1)M_1 M_1^\dag)^{-1} M_1.
\end{equation}
Hence it suffices to prove that $\|S_0+S_1\|\le 1$.
From Eq.~(\ref{norm1}) one infers  that $\|M_0\|\le 1$ and $\|M_1\|\le 1$.
Since $M_0$ has $(s-r)$ fewer rows than columns, it must have at least this many
linearly independent vectors in its nullspace. From Eq.~(\ref{norm1}) one infers that for any $\phi\in \mathbb{C}^s$ with $\|\phi\|=1$ and $M_0\phi=0$ we have $\|M_1 \phi\|=1$. Thus $M_1$ has at least $(s-r)$ singular values equal to $1$. Likewise, $M_0$ has at least $(s-r)$ singular values equal to $1$. Note that this implies that $(s-r)\le r$ since $M_0 M_0^\dagger$ is an $r\times r$ matrix with at least $(s-r)$ eigenvalues equal to $1$. Furthermore, conditions Eqs.~(\ref{norm1},\ref{norm2}) and the norm of $S_0+S_1$ are invariant under a transformation
$M_{0,1}\to WM_{0,1} V$, where $W$ and $V$ are arbitrary unitary matrices.
We can always choose $W$ and $V$ to bring $M_1$ into a diagonal form such that the diagonal matrix elements of $M_1$ are non-negative and non-increasing. Thus we can assume without loss of generality that 
\begin{equation}
\label{M1}
M_1=\left[ \ba{cc} D & 0_{r\times (s-r)}  \\
\ea \right]
\end{equation}
where
\begin{equation}
D=\mathrm{diag}(d_1,d_2,\ldots,d_{r}),
\quad 1= d_1=\ldots=d_{(s-r)}\ge d_{(s-r+1)} \ge \ldots \ge d_{r} \ge 0.
\end{equation}
(If $s=r$ the above equation should read $1\geq d_1\geq d_2\ldots \geq d_r\geq 0$.). Here and below $0_{t\times q}$ denotes an all-zeros matrix of size $t\times q$. It follows that $M_1M_1^\dag = D^2$ and thus 
\begin{equation}
\label{eq2}
S_1=\left[ \ba{cc} I_{(s-r)}\ & 0_{(s-r)\times r} \\ 0_{r\times (s-r)} & \tilde{S}_1 \\ \ea \right],
\quad \mbox{where} \quad \tilde{S}_1 = x\hat{D}^2 (I_{r}+(x-1)\hat{D}^2)^{-1},
\end{equation}
and
\begin{equation}
\label{hatD}
\hat{D}=\mathrm{diag}(d_{(s-r+1)},d_{(s-r+2)},\ldots,d_{r},0_{1\times (s-r)}).
\end{equation}
The above arguments also show that $\phi \in \CC^s$ is in the nuillspace of $M_0$ 
iff $\phi$ has support on basis vectors $i$ with $d_i=1$. 
Thus we can assume wlog that 
\begin{equation}
\label{M0}
M_0=\left[ \ba{cc} 0_{r\times (s-r)} & M \\ \ea \right],
\end{equation}
where $M$ is some matrix of size $r\times r$. Substituting Eqs.~(\ref{M1},\ref{M0}) into Eq.~(\ref{norm1},\ref{norm2}) yields
\begin{equation}
\label{MdagM}
M^\dag M =I_{r}-\hat{D}^2
\end{equation}
and
\begin{equation}
\label{MMdag}
MM^\dag \ge I_{r}-D^2.
\end{equation}
Using the polar decomposition of $M$ and Eq.~(\ref{MdagM}) we obtain the parameterization
$M=U(I_{r}-\hat{D}^2)^{1/2}$,
where $U$ is unitary. Then Eq.~(\ref{MMdag}) is equivalent
to $U(I_{r}-\hat{D}^2)U^\dag \ge I_{r} -D^2$, or 
\begin{equation}
\label{MMdag1}
U\hat{D}^2 U^\dag \le D^2.
\end{equation}
Using the definition of $S_0$ one gets
\begin{equation}
\label{eq3}
S_0=\left[ \ba{cc} 0_{(s-r)\times (s-r)} & 0_{(s-r)\times r} \\ 0_{r\times (s-r)} & \tilde{S}_0 \\ \ea \right],
\quad \mbox{where} \quad \tilde{S}_0 = M^\dag (I_{r} + (x-1) D^2)^{-1} M.
\end{equation}
Combining Eqs.~(\ref{eq2},\ref{eq3}), it suffices to show that $\|\tilde{S}_0+\tilde{S}_1\|\le 1$.
Using the chosen parameterization of $M$ one gets
\begin{equation}
\label{eq4}
\tilde{S}_0 =
 (I_{r} - \hat{D}^2)^{1/2} (I_{r}+(x-1) U^\dag D^2 U )^{-1}  (I_{r} - \hat{D}^2)^{1/2}.
\end{equation}
Now Eq.~(\ref{MMdag1}) implies $U^\dag D^2 U\ge \hat{D}^2$. 
Since $f(y)=-1/y$ is an operator monotone function, it follows that 
\begin{equation}
\label{eq5}
(I_{r}+(x-1) U^\dag D^2 U )^{-1}   \le (I_{r}+(x-1) \hat{D}^2)^{-1},
\end{equation}
that is,
\begin{equation}
\label{eq6}
\tilde{S}_0+\tilde{S}_1 \le (I_{r} - \hat{D}^2) (I_{r}+ (x-1) \hat{D}^2)^{-1} + 
 x\hat{D}^2 (I_{r}+(x-1)\hat{D}^2)^{-1}  = I.
\end{equation}
This proves that $\|\tilde{S}_0+\tilde{S}_1\|\le 1$
which implies $\|S_0+S_1\|=1$ and thus $\|R_n\|\le 1$.
\end{proof}

\begin{proof}[\bf Proof of Proposition~\ref{prop:R}]
We first show that
\begin{equation}
\label{ground1}
\calG_n^\perp =\calA + \calB, \quad \mbox{where} \quad \calA=\calG_{n-2} \otimes \psi 
\quad \mbox{and} \quad \calB=\calG_{n-1}^\perp \otimes \CC^2.
\end{equation}
Note that the two tensor products  in $\calA$ and in $\calB$ refer to two different partitions of the chain. We use the following two general properties of the orthogonal complement (here $\mathcal{W},\mathcal{V}$ are subspaces of a Hilbert space)
\begin{align}
(\mathcal{W}+\mathcal{V})^\perp&=\mathcal{W}^\perp \cap \mathcal{V}^{\perp}\label{comp1}\\
(\mathcal{W}\otimes \mathcal{V})^\perp&=\mathcal{W}^\perp \otimes \mathcal{V}+\mathcal{W} \otimes \mathcal{V}^\perp+\mathcal{W}^\perp \otimes \mathcal{V}^\perp.\label{comp2}
\end{align}
We have
\[
\calG_n=(\calG_{n-1}\otimes \CC^2)\cap (\calG_{n-2}\otimes \psi^\perp)
\]
Applying Eqs.~(\ref{comp1},\ref{comp2}) gives
\begin{align*}
\calG_n^\perp&=(\calG_{n-1}\otimes \CC^2)^\perp+ (\calG_{n-2}\otimes \psi^\perp)^\perp\\
&=\calG_{n-1}^\perp\otimes \CC^2+\calG_{n-2}^\perp \otimes \psi^\perp+\calG_{n-2}\otimes \psi+ \calG_{n-2}^\perp\otimes \psi\\
&=\calA+\calB
\end{align*}
where to get the last line we absorbed the second and fourth terms into the first, using the fact that $(G_{n-2}\otimes \CC^2)^\perp \subseteq G_{n-1}^\perp$. 

Choose an arbitrary orthonormal basis 
\begin{equation}
\label{basis3}
e_1,e_2,\ldots,e_q \in \calG_{n-1}^\perp, \quad \quad q=\dim{(\calG_{n-1}^\perp)}
\end{equation}
From Eq.~(\ref{ground1}) one infers that $\calG_n^\perp$ is spanned by
\[
(\hat{e}_1,\ldots,\hat{e}_M)=
(g_1\otimes \psi,\ldots,g_{r} \otimes \psi) \cup (e_1\otimes v_0, \ldots, e_q\otimes v_0) \cup (e_1\otimes v_1,\ldots, e_q \otimes v_1).
\]
where $v_0,v_1$ are the Schmidt vectors of $\psi$ as defined in Eq.~(\ref{psiSchmidt}).

We now show that the Gram matrix $\Gamma$ defined by
\[
\Gamma_{p,q}=\langle \hat{e}_p|\hat{e}_q\rangle
\]
is invertible, which implies that $\hat{e}_1,\ldots,\hat{e}_M$ are linearly independent. We note that $\Gamma$ has the following simple form
\begin{equation}
\label{Gamma}
\Gamma=\left[ 
\ba{ccc}
I_r & B_0 & B_1 \\
B_0^\dag & I_q & 0 \\
B_1^\dag & 0 & I_q \\
\ea
\right]
\end{equation}
where
\begin{equation}
\label{Bdef}
(B_z)_{\alpha,i} = \langle g_\alpha \otimes \psi| e_i\otimes v_z\rangle= \sqrt{p_z} \langle g_\alpha \otimes w_z| e_i\rangle \qquad z=0,1.
\end{equation}
Define $B=\left[\ba{cc} B_0 & B_1\ea \right]$, $X=BB^\dagger$, and $Y=B^\dagger B$. Note that $X$ and $Y$ have the same non-zero eigenvalues. Also note that $\Gamma$ is invertible if none of these eigenvalues are equal to $1$, since in this case 
\begin{equation}
\label{invGamma}
\Gamma^{-1}=\left[ 
\ba{cc}
(I_r-X)^{-1} &  -B(I_{2q}-Y)^{-1}\\
-B^\dagger(I_r-X)^{-1} & (I_{2q}-Y)^{-1}\\
\ea
\right].
\end{equation}
To show that $\Gamma$ is invertible it therefore suffices to show that $I_r-X$ is invertible. Using Eqs.~(\ref{gram},\ref{Bdef}) 
and the identity $I=G_{n-1}+G_{n-1}^\perp$
we get
\[
\frac{1}{p_0}B_0B_0^\dagger+M_0M_0^\dagger=I_r \quad \text{and} \quad 
\frac{1}{p_1}B_1B_1^\dagger+M_1M_1^\dagger=I_r.
\]
So 
\begin{equation}
\label{IminusX}
I_r-X=I_r-B_0B_0^\dagger-B_1B_1^\dagger=p_0M_0M_0^\dagger+p_1M_1M_1^\dagger.
\end{equation}
To prove that $I_r-X$ is invertible we show that this operator is positive definite:
\[
p_0M_0M_0^\dagger+p_1M_1M_1^\dagger\geq \min{(p_0,p_1)}(M_0M_0^\dagger+M_1M_1^\dagger)\geq \min{(p_0,p_1)}I >0
\]
where we used Eq.~(\ref{norm2}) and the fact that $p_0,p_1$ are both positive. This completes the proof that $\Gamma$ is invertible and establishes that $\hat{e}_1,\ldots \hat{e}_M$ are linearly independent.

Since we have shown that $\hat{e}_1,\ldots \hat{e}_M$ are a basis for $\calG_n^\perp$, we have
\begin{equation}
\label{ground3}
G_n^\perp=\sum_{p,q=1}^M (\Gamma^{-1})_{p,q} |\hat{e}_p\rangle\langle \hat{e}_q|.
\end{equation}

Substituting Eqs.~(\ref{ground3},\ref{invGamma}) into 
$R_n\equiv G_{n-1} \trace_n(G_n^\perp) G_{n-1}$
and noting that $\hat{e}_\alpha=g_\alpha\otimes \psi$ with $\alpha=1,\ldots,r$ 
are the only basis vectors of $\calG_n^\perp$ which are not orthogonal
to $\calG_{n-1} \otimes \CC^2$, we arrive at
\[
R_n=\sum_{\alpha,\beta=1}^{r} (I_r-X)^{-1}_{\alpha,\beta} \, G_{n-1} \left(|g_\alpha\rangle\langle g_\beta| \otimes \trace_2|\psi\rangle\langle \psi| )\right)G_{n-1}.
\]
Substituting $G_{n-1}=\sum_{i=1}^s |h_i\rangle\langle h_i|$ 
and $\trace_2|\psi\rangle\langle \psi|=p_0|w_0\rangle\langle w_0|+p_1 |w_1\rangle\langle w_1|$ 
into the above equation
yields
\[
\langle h_i|R_n|h_j\rangle = \sum_{\alpha,\beta=1}^{r}  (I_r-X)^{-1}_{\alpha,\beta}
\left(p_0\langle h_i|g_\alpha \otimes w_0\rangle \cdot \langle g_\beta \otimes w_0|h_j\rangle+p_1\langle h_i|g_\alpha \otimes w_1\rangle \cdot \langle g_\beta \otimes w_1|h_j\rangle\right).
\]
Replacing the last two factors by matrix elements of  $M_0,M_1$ defined in Eq.~(\ref{gram}) and $I_r-X$ by Eq.~(\ref{IminusX}) 
 one arrives at Eq.~(\ref{Rsimple1}). 
\end{proof}


\section{Gapped phase}
\label{sec:gapped}

In this section we prove the second part of Theorem~\ref{thm:main}, namely
\begin{gappedphasethm}
Suppose  the eigenvalues of $T_\psi$ have different magnitudes or 
both eigenvalues are equal to zero. Then
the spectral gap of $H_n(\psi)$ is lower bounded by a positive constant independent of $n$. 
\end{gappedphasethm}

Let us first consider the simple case when
both eigenvalues of $T_\psi$ are equal to zero. 
Using the canonical form of $\psi$ established  in Proposition~\ref{prop:canonical} 
and Eq.~(\ref{transfer}) one can check that this is possible only if 
$|\psi\rangle=(U\otimes U)|1,1\rangle$ for some single-qubit unitary operator $U$.
Thus the Hamiltonian $H_n(\psi)$ is a sum of pairwise commuting projectors
and $\gamma(\psi,n)\ge 1$ for all $n$ which proves the desired lower bound.

In the rest of this section we assume that  the eigenvalues of $T_\psi$ have distinct magnitudes.
In this case the eigenvectors of $T_\psi$ must be  linearly independent. Let us first introduce some notation. Suppose $S\subseteq [n]$ is a consecutive block of qubits. We shall write $G_S$ for the projector onto the ground space of the truncated  Hamiltonian 
\[
\sum_{\{i,i+1\} \subseteq S} |\psi\rangle\langle \psi|_{i,i+1}
\]
obtained from $H_n(\psi)$ by retaining only 
the terms fully contained in $S$. The projector $G_S$ acts trivially on all qubits in the complement of $S$.
Note that $G_n=G_S$ in the case where $S$ is the entire chain. 

Our starting point  is a general lower bound on the gap of 1D frustration-free Hamiltonians  due to Nachtergaele~\cite{Nachtergaele1996}.
Specializing Theorem~3 of Ref.~\cite{Nachtergaele1996} to our case one gets the following lemma.
\begin{Lemma}[Nachtergaele~\cite{Nachtergaele1996}]
\label{lemma:N}
Suppose there exists an integer $r\ge 1$ and a real number $\epsilon<(r+1)^{-1/2}$
such that for all large enough $n$ and for the partition $[n]=ABC$ with $|B|= r$ and $|C|=1$ one has
 $\|G_{ABC}- G_{AB} G_{BC}\|\le \epsilon$.
  Then 
\begin{equation}
\label{Delta_n}
\gamma(\psi,n) \ge \frac{\gamma(\psi,r+1)}{r+1}\left(1-\epsilon (r+1)^{1/2}\right)^2
\end{equation}
for all large enough $n$.
\end{Lemma}

Note that the right-hand side of Eq.~(\ref{Delta_n}) is a positive constant independent of $n$.
Thus Lemma~\ref{lemma:N} 
reduces the problem of lower bounding the spectral gap of $H_n(\psi)$ to that
of upper bounding the  quantity  $\|G_{ABC}- G_{AB} G_{BC}\|$.
Our main technical result is an upper bound on this quantity 
that decays exponentially with the size of $B$. 
\begin{theorem}
\label{thm:Nbound}
Let $\mu_1,\mu_2$ be the eigenvalues of $T_\psi$ such that $|\mu_1|<|\mu_2|$.
Define $\lambda=\mu_2/\mu_1$.
Let $c$ be the inner product between the normalized eigenvectors of $T_\psi$.
Consider any partition $[n]=ABC$ such that $|B|=r$. Then
\begin{equation}
\label{Nbound}
\| G_{ABC}-G_{AB}G_{BC}\|\le O\left( r^{1/2} |\lambda|^{-r/8}\right) + O\left(|c|^{r/8}\right),
\end{equation}
where the constant coefficients in $O(\cdot)$ depend only on the forbidden state $\psi$.
If $\mu_1=0$ then \eqref{Nbound} holds with a formal replacement $\lambda=\infty$ which sets the first term to zero.
\end{theorem}

Note that $|c|<1$ since the eigenvectors of $T_\psi$ are
linearly independent. Furthermore, since $|\lambda|>1$, the right-hand side of Eq.~(\ref{Nbound})
is an exponentially decaying function of $r$.  
Therefore we can choose a constant $r$ depending only on the forbidden
state $\psi$ such that 
\[
\|G_{ABC}-G_{AB}G_{BC}\| \le \epsilon\equiv  \frac1{2(r+1)^{1/2}}
\]
for all $n>r$. Substituting this into Lemma~\ref{lemma:N} one gets
\[
\gamma(\psi,n)\ge \frac{\gamma(\psi,r+1)}{4(r+1)}
\]
for all large enough $n$ which proves the gapped phase theorem. 

In the rest of this section we prove Theorem~\ref{thm:Nbound}. 
We shall first consider the case where $\psi$ is an entangled state ($\mu_1\ne 0$).
The main technical difficulty we had to overcome is 
a lack of  explicit formulas for the projectors
$G_{ABC}, G_{AB}$, and $G_{BC}$.
At a high level, our approach is to develop a set of identities 
relating the global ground space projector such as $G_{ABC}$
and the local ones such as $G_A$ or $G_{AB}$. 
These identities hold with a small error  controlled by the size of the regions. 
Our proof of the theorem uses three identities of this type 
which are stated as ``Region Exclusion" lemmas in Section~\ref{subs:regionexclusion}.
We use the Region Exclusion lemmas to decompose the operator $G_{ABC}-G_{AB}G_{BC}$
in Eq.~(\ref{Nbound}) into a sum of several terms and to show that the norm of each term
is exponentially small in $r$, see Section~\ref{subs:finalgappedproof}.
The proof of the Region Exclusion lemmas combines two  ingredients: monotonicity of the ground space projectors under the partial trace (established in Section~\ref{sec:mono}) and the fact that certain correlation functions in the ground space decay exponentially (established in Section~\ref{subs:corrfunc}).

In Section~\ref{subs:simple} we consider the case where $\psi$ is a product state ($\mu_1=0$). In this case the orthonormal basis for the ground space constructed in Section~\ref{subs:open}
gives an explicit formula for the ground space projector. We use this formula to establish the Region Exclusion lemmas (for the $\mu_1=0$ case) in a more direct way. The corresponding special case of the theorem then follows from the Region Exclusion identities.

Before proceeding, let us establish some notation and conventions. 
Recall that a local unitary transformation $\psi\to (U\otimes U)\psi$
preserves  eigenvalues of $H_n(\psi)$ and maps $T_\psi$
to $(\det{U})^{-1} \cdot U T_\psi U^\dag$, see Section~\ref{sec:gapless}.
We shall choose the unitary $U$ to fix one of the eigenvectors
of $T_\psi$. Specifically, in the rest of this section we shall assume that  
\begin{equation}
\label{Tevecs}
T_\psi |0\rangle=\mu_1 |0\rangle \quad \text{and} \quad T_\psi |v\rangle=\mu_2 |v\rangle
\end{equation}
for some  state 
\[
|v\rangle=c|0\rangle + s|1\rangle, \quad \mbox{where $|c|^2+s^2=1$ and $s>0$}.
\]
Note that $c$ is the inner product between the eigenvectors of $T_\psi$, as defined
in the statement of Theorem~\ref{thm:Nbound}.
It will also be convenient to define a state
\[
|v^\perp\rangle=s|0\rangle-c^*|1\rangle.
\]
Given a set of qubits $S$ and a single-qubit state $|\theta\rangle$ we shall write $|\theta \rangle_S$ for the
 product state $|\theta \rangle^{\otimes |S|}$. We shall write $|\theta\rangle\langle\theta|_S$ for the projector onto this state and $|\theta\rangle\langle\theta|_S^\perp=I_S-|\theta\rangle\langle\theta|_S$.

\subsection{Correlation functions}
\label{subs:corrfunc}
In this section we show that certain ground space correlation functions decay exponentially. Specifically, define
\begin{equation}
\label{tau2point}
\tau(i,j,n)=\trace{\left(G_n |1\rangle\langle 1|_i\otimes |v^{\perp}\rangle\langle v^{\perp}|_j\right)},
\end{equation}
and
\begin{equation}
\label{tau1point}
\tau(n)=\trace{\left( G_n  |v^\perp\rangle\langle v^\perp|_n\right)}
\end{equation}
where $1\leq i<j\leq n$. For notational convenience we have suppressed the dependence of these functions on the forbidden state $\psi$. 
Our main result in this section is as follows.
\begin{Lemma}\label{lemma:decay}
The sequence $\{\tau(n)\}_{n\ge 2}$ is monotonically increasing
and has a finite limit  $\tau(\infty)$ such that
\begin{equation}
\label{limit}
0\le \tau(\infty)-\tau(n) \le O\left(n|\lambda|^{-2n}\right) \quad \quad \mbox{for all $n\ge 2$}.
\end{equation} 
Furthermore,
\begin{equation}
\label{decay}
\tau(i,j,n)\le O\left((j-i) \cdot |\lambda|^{-2(j-i)}\right) \quad \quad \mbox{for all $1\leq i<j\leq n$}.
\end{equation}
Here the constant coefficients in $O(\cdot)$ depend only on the forbidden state $\psi$.
\end{Lemma}
\begin{proof}
Let us define yet another correlation function 
\begin{equation}
\sigma(i,j,n)=\max_{\phi\in \calG_n} \; \langle\phi|\left(|1\rangle\langle1|_{i}\otimes|v^{\perp}\rangle\langle v^{\perp}|_{j}\right)|\phi\rangle,
\label{Kdef}
\end{equation}
where the maximum is over normalized ground states, that is, $\|\phi\|=1$.
\begin{prop}\label{prop:easyfact}
\begin{equation}
\sigma(i,j,n) \leq|\lambda|^{-2(j-i)}\frac{s^{2}}{1-|c|}
\end{equation}
for $1\leq i<j\leq n$.
\end{prop}
\begin{proof}
Define
\[
|\psi^{r}\rangle=|v^{\perp}1\rangle- (\lambda^*)^r |1v^{\perp}\rangle
\]
Using Eq.~\eqref{Tevecs} we see that
$\langle 1| T_\psi =\mu_2 \langle 1|$ and $\langle v^\perp|T_\psi =\mu_1 \langle v^\perp|$.
Therefore 
\begin{equation}
\label{psir}
\langle \psi^r |(I\otimes T_\psi^r)\sim \mu_2^r \langle v^\perp 1| - (\lambda \mu_1)^r \langle 1v^\perp|
\sim \langle v^\perp 1| -  \langle 1v^\perp| \sim \langle \epsilon|,
\end{equation}
where $|\epsilon\rangle =|0,1\rangle-|1,0\rangle$
(recall that $\sim$ means proportional to). Comparing Eq.~(\ref{Teps}) and Eq.~\eqref{psir} 
one infers that $\psi^1$ is  the forbidden state, that is, $\psi \sim \psi^1$.

Let $\phi\in \mathcal{G}_n$ be a normalized state (i.e., $\left\Vert \phi\right\Vert=1$) for which the maximum in
 Eq.~\eqref{Kdef} is achieved, so 
\[
\sigma(i,j,n)= \langle\phi|\left(|1\rangle\langle1|_{i}\otimes|v^{\perp}\rangle\langle v^{\perp}|_{j}\right)|\phi\rangle.
\]
Since $\phi\in \mathcal{G}_n$, by Proposition \ref{gsdegen} it can be written $|\phi\rangle=I\otimes T_\psi \otimes T^{2}_\psi \ldots\otimes T^{n-1}_\psi|\chi\rangle$ where $|\chi\rangle$ belongs to
the symmetric subspace. Using this fact and Eq.~\eqref{psir} we see that
\begin{equation}
\label{ijpsi}
_{i,j}\langle\psi^{j-i}|\phi\rangle=0
\end{equation}
for all integers $1\le i<j\le n$.
Writing
\[
|\phi\rangle=|0\rangle_{i}|\phi_{0}^{i}\rangle_{[n]\setminus{i}}+|v\rangle_{i}|\phi_{1}^{i}\rangle_{[n]\setminus{i}}
\]
and substituting this into Eq.~(\ref{ijpsi}) one gets
\[
s\left(_{j}\langle1|\phi_{0}^{i}\rangle-\lambda^{j-i}{}_{j}\langle v^{\perp}|\phi_{1}^{i}\rangle\right)=0
\]
which implies
\begin{equation}
\langle\phi_{1}^{i}|\left(|v^{\perp}\rangle\langle v^{\perp}|_{j}\right)|\phi_{1}^{i}\rangle=\frac{1}{|\lambda|^{2(j-i)}}\langle\phi_{0}^{i}|\left(|1\rangle\langle1|_{j}\right)|\phi_{0}^{i}\rangle.\label{eq:implication}
\end{equation}
Using the fact that $|\phi\rangle$ is normalized we have
\begin{equation}
1=\langle\phi|\phi\rangle=\langle\phi_{0}^{i}|\phi_{0}^{i}\rangle+\langle\phi_{1}^{i}|\phi_{1}^{i}\rangle+2\mathrm{Re}\left(c\langle\phi_{0}^{i}|\phi_{1}^{i}\rangle\right).\label{eq:normalized}
\end{equation}
We upper bound the magnitude of the third term using the Cauchy-Schwarz and the arithmetic/geometric mean inequality:
\[
\left|c\langle\phi_{0}^{i}|\phi_{1}^{i}\rangle\right|\leq|c|\sqrt{\langle\phi_{0}^{i}|\phi_{0}^{i}\rangle\langle\phi_{1}^{i}|\phi_{1}^{i}\rangle}\leq\frac{|c|}{2}\left(\langle\phi_{0}^{i}|\phi_{0}^{i}\rangle+\langle\phi_{1}^{i}|\phi_{1}^{i}\rangle\right).
\]
Substituting this into Eq.~(\ref{eq:normalized}) yields
\[
1\geq\left(1-|c|\right)\left(\langle\phi_{0}^{i}|\phi_{0}^{i}\rangle+\langle\phi_{1}^{i}|\phi_{1}^{i}\rangle\right)\geq\left(1-|c|\right)\langle\phi_{0}^{i}|\phi_{0}^{i}\rangle
\]
 and hence $\langle\phi_{0}^{i}|\phi_{0}^{i}\rangle\leq\frac{1}{1-|c|}$.
Using this fact and Eq.~(\ref{eq:implication}) we obtain
\[
\langle\phi_{1}^{i}|\left(|v^{\perp}\rangle\langle v^{\perp}|_{j}\right)|\phi_{1}^{i}\rangle\leq\frac{|\lambda|^{-2(j-i)}}{1-|c|}
\]
and thus 
\[
\sigma(i,j,n)=\langle\phi|\left(|1\rangle\langle1|_{i}\otimes|v^{\perp}\rangle\langle v^{\perp}|_{j}\right)|\phi\rangle=s^{2}\langle\phi_{1}^{i}|\left(|v^{\perp}\rangle\langle v^{\perp}|_{j}\right)|\phi_{1}^{i}\rangle\leq|\lambda|^{-2(j-i)}\frac{s^{2}}{1-|c|}.
\]
\end{proof}
\noindent
Now we are ready to prove Eq.~(\ref{limit}).
First, applying the Monotonicity Lemma (Lemma \ref{lemma:mono}) to the
left-right flipped chain  yields 
$\trace{}_1(G_n)\ge G_{n-1}$. Therefore 
\[
\tau(n)=\trace{\left( \trace{}_1{(G_n)}  |v^\perp\rangle\langle v^\perp|_{n-1} \right)} 
\ge \trace{\left( G_{n-1}  |v^\perp\rangle\langle v^\perp|_{n-1}\right)} =\tau(n-1),
\]
that is, $\tau(n)$ is monotonically increasing.

Inserting the identity decomposition $I=|0\rangle\langle 0|+|1\rangle\langle 1|$
on the first qubit in Eq.~(\ref{tau1point}) one gets
\begin{equation}
\label{tau01}
\tau(n)=\trace{(G_n |1\rangle\langle 1|_1 \otimes  |v^\perp\rangle\langle v^\perp|_n)}
+ \trace{(G_n |0\rangle\langle 0|_1 \otimes  |v^\perp\rangle\langle v^\perp|_n)}.
\end{equation}
The first term in Eq.~(\ref{tau01}) is upper bounded by
$(n+1)\sigma(1,n,n)$ since we can decompose 
$G_n=\sum_{a=0}^n |g_a\rangle\langle g_a|$ using some 
orthonormal basis $\{g_a\}$  of $\calG_n$ and use the fact that
\[
\langle g_a|(|1\rangle\langle 1|_1 \otimes  |v^\perp\rangle\langle v^\perp|_n)|g_a\rangle
\le \sigma(1,n,n)
\]
for each individual state $g_a$. 
The second term in Eq.~(\ref{tau01}) is upper bounded by 
$\tau(n-1)$, which follows from $G_n\le I \otimes  G_{n-1}$.
Thus 
\[
\tau(n)
\le (n+1) \sigma(1,n,n) + \tau(n-1).
\] 
Proposition~\ref{prop:easyfact} implies
$\sigma(1,n,n)=O(|\lambda|^{-2n})$, that is,
\[
0\le \tau(n)-\tau(n-1)\le O(n|\lambda|^{-2n}).
\]
This shows that $\tau(n)$ has a finite limit $\tau(\infty)$ at $n\to \infty$. 
Summing up the series produces the desired bound Eq.~(\ref{limit}).

The proof of  Eq.~(\ref{decay}) follows a similar strategy. First consider the case $i=1,j=n$. The same argument used above shows that 
\begin{equation}
\label{11n}
\tau(1,n,n)\le (n+1)\sigma(1,n,n)=O(n|\lambda|^{-2n}).
\end{equation}

Next suppose $i\geq 2$ and $j=n$. Inserting the identity decomposition $I=|0\rangle\langle0|+|1\rangle\langle 1|$
on the first qubit, using the fact that $G_n\le I\otimes G_{n-1}$, and noting that $|1\rangle\langle 1|_i\le I$, one gets
\begin{eqnarray}
\tau(i,n,n) & \le &  \tau(i-1,n-1,n-1) + 
 \trace{(G_n |1\rangle\langle 1|_1 \otimes  |1\rangle\langle 1|_i \otimes 
|v^\perp\rangle\langle v^\perp|_n )} \nonumber \\
&\le& \tau(i-1,n-1,n-1) + \tau(1,n,n).
\end{eqnarray}
This shows that 
\begin{equation}
\label{inn}
\tau(i,n,n)\le \sum_{k=n-i+1}^n \tau(1,k,k).
\end{equation}
Substituting Eq.~(\ref{11n}) into this bound and summing up the series, we get
\begin{equation}
\tau(i,n,n)\le O\left( (n-i)|\lambda|^{-2(n-i)}\right) \quad \text{for all} \quad i=1,\ldots,n-1.
\end{equation}
(Here we included the case $i=1$ which was handled in  Eq.~(\ref{11n})).

Finally, consider the case $j\leq n-1$.  Inserting the identity decomposition $I=|v\rangle\langle v|+|v^\perp\rangle\langle v^\perp|$ on the $n$-th qubit, using the fact that $G_n\le G_{n-1}\otimes I$, and noting that $|v^\perp\rangle\langle v^\perp|_j\le I$,  one gets
\[
\tau(i,j,n)\le \tau(i,j,n-1) + \trace{(G_n |1\rangle\langle 1|_i \otimes 
|v^\perp\rangle\langle v^\perp|_j \otimes |v^\perp\rangle\langle v^\perp|_n )}
\le \tau(i,j,n-1)+ \tau(i,n,n).
\]
This shows that 
\begin{equation}
\label{ijn1}
\tau(i,j,n) \le \sum_{m=j}^n \tau(i,m,m).
\end{equation}
Combining this with Eq.~(\ref{inn}) leads to the desired bound Eq.~(\ref{decay}).
\end{proof}

\subsection{Region exclusion lemmas}
\label{subs:regionexclusion}
To perform manipulations with ground space projectors
that involve several regions we now prove three region exclusion lemmas. These lemmas
allow one to exclude one of the regions from certain operators built from ground space projectors.

The first region exclusion lemma 
states that $G_{ABC} |v\rangle\langle v|_{BC}\approx G_{AB}\otimes I_C |v\rangle\langle v|_{BC}$.

\setcounter{regionexclusion}{4}
\begin{regionexclusion}
\label{lemma:C1}
Let $[n]=ABC$ with $|B|=j$. Then
\[
\left\Vert \left(G_{ABC}-G_{AB}\otimes I_C\right)|v\rangle\langle v|_{BC}\right\Vert^2 \le O(|c|^j)+O(j|\lambda|^{-j})
\]
Here the constant coefficients in $O(\cdot)$ depend only on the forbidden state $\psi$.
\end{regionexclusion}
\begin{proof}
Define  $P\equiv G_{ABC}$ and $Q\equiv G_{AB}\otimes I_C$. Using the fact that $PQ=QP=P$ one can write the quantity we wish to bound as
\begin{equation}
\left\Vert \left(P-Q\right)|v\rangle\langle v|_{BC}\right\Vert^2 
\le  \trace{}_A \langle v_{BC} |(P-Q)^2 |v_{BC}\rangle 
=\trace{(Q |v\rangle\langle v|_{BC})} - \trace{(P|v\rangle\langle v|_{BC})}.
\label{X1eq2}
\end{equation}
Define
\begin{equation}
\label{X1eq3}
\theta(n,r)=\trace{( G_n \cdot I_{n-r} \otimes |v\rangle\langle v|^{\otimes r})}.
\end{equation}
One can rewrite Eq.~(\ref{X1eq2})  as
\begin{equation}
\label{X1eq4}
\left\Vert \left(P-Q\right)|v\rangle\langle v|_{BC}\right\Vert^2 \le \theta(i+j,j)-\theta(i+j+k,j+k),
\end{equation}
where $i=|A|$, $j=|B|$, and $k=|C|$.
Representing $|v\rangle\langle v|=I-|v^\perp\rangle\langle v^\perp|$ on the last qubit
in Eq.~(\ref{X1eq3})
and using the monotonicity property $\trace{}_n(G_n)\ge G_{n-1}$ from Lemma \ref{lemma:mono} one gets
\begin{equation}
\label{X1eq8}
\theta(n,r)\ge \theta(n-1,r-1) -\xi(n,r-1),
\end{equation}
where
\begin{equation}
\label{X1eq9}
 \xi(n,r)\equiv \trace{( G_n \cdot I_{n-r-1} \otimes |v\rangle\langle v|^{\otimes r}\otimes |v^\perp\rangle \langle 
v^\perp|)}.
\end{equation}
From Eq.~\eqref{X1eq8} we get
\[
\theta(i+j,j)\leq \theta(i+j+k,j+k)+\sum_{r=j}^{j+k-1}  \xi(i+r+1,r)
\]
and plugging into Eq.~\eqref{X1eq4} gives
\begin{equation}
\left\Vert \left(P-Q\right)|v\rangle\langle v|_{BC}\right\Vert^2\leq \sum_{r=j}^{j+k-1}  \xi(i+r+1,r).
\label{eps2bnd}
\end{equation}

To complete the proof we now show that $\xi(n,r)$ 
has an upper bound which is exponentially small in $r$
and does not depend on $n$. Partition the chain as $[n]=A'B'B''C'$, where
$|C'|=1$, $|B'|=\lfloor r/2\rfloor$, $|B''|=\lceil r/2\rceil$, and $|A'|=n-1-r$. Using the fact that $|v\rangle\langle v|_{B''}\leq I$ we get
\begin{equation}
\label{X1eq10}
\xi(n,r)\le \langle v_{B'} | L_{B'} |v_{B'}\rangle, \quad  \mbox{where}
\quad L_{B'} \equiv \trace{}_{A'B''C'}( G_n |v^\perp\rangle\langle v^\perp|_{C'}).
\end{equation}
Using the second part of Lemma~\ref{lemma:decay} 
 we have
\begin{equation}
\label{X1eq11}
\trace{(L_{B'} |1\rangle\langle 1|_m )}=\tau(m,n,n) = O( (n-m) |\lambda|^{-2(n-m)})
\quad \mbox{for any $m\in B'$}.
\end{equation}
Note that $n-m\ge r/2$ for any $m\in B'$.
Let $|0\rangle\langle 0|_{B'}^\perp = I -|0\rangle\langle 0|_{B'}$. 
It follows that 
\begin{equation}
\label{X1eq12}
\trace{(L_{B'} |0\rangle\langle 0|_{B'}^\perp)} \le \sum_{m\in B'} \trace{(L_{B'} |1\rangle\langle 1|_m )} 
\le O(1) \sum_{m=n-r}^{n-r/2} (n-m)|\lambda|^{-2|n-m|}
= O(r |\lambda|^{-r}).
\end{equation}
Thus $L_{B'}$ has almost all its weight on the basis vector $|0_{B'}\rangle$
and an exponentially small weight $O(r |\lambda|^{-r})$ on the space orthogonal to $|0_{B'}\rangle$.
Furthermore, the first part of Lemma~\ref{lemma:decay}
implies that $\trace{(L_{B'})}=\tau(n)=\tau(\infty)-O(n|\lambda|^{-2n})$.
Combining this fact and Eq.~(\ref{X1eq12}) results in
\begin{equation}
\label{X1eq13}
L_{B'} = \tau(\infty) |0\rangle\langle 0|_{B'} + E \quad \text{where} \quad \|E\|\leq O(r |\lambda|^{-r}).
\end{equation}
Therefore
\begin{equation}
\label{X1eq14}
\xi(n,r)\le \langle v_{B'}|L_{B'}|v_{B'}\rangle =\tau(\infty) |\langle 0|v\rangle|^{2|B'|}+O(r |\lambda|^{-r}) \leq \tau(\infty) |c|^{(r-2)}
+O(r |\lambda|^{-r}).
\end{equation}
Finally, substituting this into Eq.~\eqref{eps2bnd} gives
\begin{equation}
\label{X1eq15}
\left\Vert \left(P-Q\right)|v\rangle\langle v|_{BC}\right\Vert^2 \le \sum_{r=j}^{j+k-1} \xi(i+r+1,r) \le \sum_{r=j}^\infty \xi(i+r+1,r) \le
 O(|c|^j) + O(j |\lambda|^{-j}).
\end{equation}
\end{proof}

The second region exclusion lemma concerns the operator $G_{ABCD}|v\rangle\langle v|_C^\perp$ (recall that $|v\rangle\langle v|_C^\perp=I-|v\rangle\langle v|_C$).

\begin{regionexclusion}
\label{lemma:C2}
Consider any partition $[n]=ABCD$ with $|B|=j$. Then
\[
\| \left(G_{ABCD} - |0\rangle\langle 0|_A \otimes G_{BCD} \right)|v\rangle\langle v|_C^\perp \|^2 \le O(j |\lambda|^{-2j}).
\]
Here the constant coefficient in $O(\cdot)$ depends only on the forbidden state $\psi$.
\end{regionexclusion}
\begin{proof}
For brevity 
denote $P\equiv G_{ABCD}$ and $Q\equiv |0\rangle\langle 0|_A \otimes G_{BCD}$. Then
\[
\| (P - Q) |v\rangle\langle v|_C^\perp  \|^2 \le 
 \trace{\left( (P-Q)^2 |v\rangle\langle v|_C^\perp \right)}.
\]
 Taking into account that $G_{ABCD}\, G_{BCD}=G_{ABCD}$ gives
\[
\| (P - Q) |v\rangle\langle v|_C^\perp  \|^2  \le \trace{\left( P |v\rangle \langle v|_C^\perp \right)}
 +   \trace{\left( Q |v\rangle \langle v|_C^\perp \right)} \nonumber\\
- 2\trace{\left( P |0\rangle\langle0|_A \otimes |v\rangle \langle v|_C^\perp\right)}.
\]
Substituting $|0\rangle\langle0|_A=I- |0\rangle\langle0|_A^\perp$ in the last term results in 
\[
\| (P - Q) |v\rangle\langle v|_C^\perp  \|^2  \le 
 \trace{\left( Q |v\rangle \langle v|_C^\perp \right)}
-\trace{\left( P |v\rangle \langle v|_C^\perp \right)} 
+2\trace{\left( P |0\rangle\langle0|_A^\perp \otimes  |v\rangle \langle v|_C^\perp\right)}.
\]
Applying the Monotonicity Lemma (Lemma~\ref{lemma:mono}) one gets
$\trace{}_A{(G_{ABCD})}\ge G_{BCD}$.  This shows that 
\[
 \trace{\left( Q |v\rangle \langle v|_C^\perp \right)}
-\trace{\left( P |v\rangle \langle v|_C^\perp \right)} = \trace{(G_{BCD}  |v\rangle \langle v|_C^\perp )}
- \trace{(G_{ABCD}  |v\rangle \langle v|_C^\perp )} \le 0
\]
and therefore
\[
\| (P - Q) |v\rangle\langle v|_C^\perp  \|^2  \le 
2\trace{\left( P |0\rangle\langle0|_A^\perp \otimes  |v\rangle \langle v|_C^\perp\right)}
\le 2\sum_{m\in A} \sum_{m'\in C} \trace{\left( G_{ABCD}
|1\rangle\langle 1|_m \otimes |v^\perp\rangle\langle v^\perp|_{m'} \right)}.
\]
One can recognize the last term as the correlation function $\tau(m,m',n)$ defined
in Section~\ref{subs:corrfunc}. Using the second part of Lemma~\ref{lemma:decay}
one gets
\[
\| (P - Q) |v\rangle\langle v|_C^\perp  \|^2  \le 
\sum_{m\in A} \sum_{m'\in C} O\left((m'-m)|\lambda|^{-2(m'-m)} \right)
\le O(1)\cdot \sum_{r=j}^\infty r(r-j) |\lambda|^{-2r}
=O(j |\lambda|^{-2j}).
\]
Here we denoted $r=m'-m$ so that $r\ge |B|=j$. We also used the fact that the number
of pairs $(m,m')$ with $m\in A$ and $m'\in C$ such that $m'-m=r$
is at most $r-|B|=r-j$.
\end{proof}

The third and final region exclusion lemma involves operators built from the ground space projectors as follows. Given any bipartition $[n]=AB$, where $A$ and $B$ are consecutive blocks of qubits, define
\[
G_{A>B}\equiv G_A \otimes |v\rangle\langle v|_B - G_{AB}
\quad \mbox{and} \quad G_{A<B} \equiv |0\rangle\langle 0|_A \otimes G_B - G_{AB}.
\]
\begin{regionexclusion}
\label{lemma:X}
Consider any partition $[n]=ABC$ with $|B|=j$. Then 
\begin{equation}
\label{X1}
\|  G_{AB>C}  - |0\rangle\langle 0|_A  \otimes G_{B>C} \| \le O\left(j^{1/2} |\lambda|^{-j/4}\right) + 
O\left( |c|^{j/4}\right).
\end{equation}
and
\begin{equation}
\label{X2}
\|  G_{A<BC}  - G_{A<B} \otimes |v\rangle\langle v|_C \| \le  O\left(j^{1/2} |\lambda|^{-j/4}\right) + 
O\left( |c|^{j/4}\right).
\end{equation}
Here the constant coefficients in $O(\cdot)$ depend only on the forbidden state $\psi$.
\end{regionexclusion}
\begin{proof}
We first show that the bound  Eq.~(\ref{X2}) follows from  Eq.~(\ref{X1}) and thus it suffices to prove the latter. To see this, consider horizontally flipping the chain so that the vertices previously labeled $1,2,\ldots,n$ are now $n,n-1,\ldots,1$. The new forbidden state is $\psi'= \text{SWAP}\psi$ where $\text{SWAP}$ is the unitary transformation which interchanges the two qubits. The new  matrix $T_{\psi'}$ is proportional to $T_\psi^{-1}$. From this we see that $|0'\rangle =|v\rangle$ and $|v'
\rangle= |0\rangle$, and that $\lambda'=\lambda$. Using these facts we can see that Eq.~(\ref{X2}) is just  Eq.~(\ref{X1}) applied to the left-right flipped chain.

It remains to prove  Eq.~(\ref{X1}). Let $\phi$ be any normalized state of $ABC$ such that 
\begin{equation}
\label{Xstep1}
\|  G_{AB>C}  - |0\rangle\langle 0|_A  \otimes G_{B>C} \| = \| (G_{AB>C}  - |0\rangle\langle 0|_A  \otimes G_{B>C})\phi\|.
\end{equation}
Partition the region $B$ as $B=B'B''$, where $|B'|=\lfloor{j/2}\rfloor$ and $|B''|=\lceil{j/2}\rceil$.
Define states
\[
\phi^{-}= |v\rangle\langle v|_C^\perp \cdot \phi, \quad 
\quad
\phi^{-+}=  |v\rangle\langle v|_{B''}^\perp \otimes |v\rangle\langle v|_C \cdot \phi,
\quad 
\quad
\phi^{++}=|v\rangle\langle v|_{B''} \otimes |v\rangle\langle v|_C \cdot \phi.
\]
One can easily check that the above states are pairwise orthogonal, 
\[
\phi=\phi^- +\phi^{-+} + \phi^{++} \quad \mbox{and} \quad 1=\|\phi\|^2=\| \phi^-\|^2 + \| \phi^{-+}\|^2 + \|\phi^{++}\|^2.
\]
We shall bound the contributions to the right-hand side of Eq.~(\ref{Xstep1}) coming from 
$\phi^-,\phi^{-+}$, and $\phi^{++}$ separately.

Let us start with $\phi^-$. Using the definitions of $G_{AB>C}$ and $G_{B>C}$ one gets
\[
G_{AB>C} \cdot \phi^-=-G_{ABC} \cdot \phi^- \quad \mbox{and} \quad |0\rangle\langle 0|_A \otimes G_{B>C} \cdot \phi^-
=-|0\rangle\langle0|_A\otimes G_{BC} \cdot \phi^-.
\]
Since $\phi^-=|v\rangle\langle v|_C^\perp\cdot \phi^-$ and $\|\phi^{-}\|\le 1$, this results in
\begin{equation}
\label{Xstep2}
 \| (G_{AB>C}  - |0\rangle\langle 0|_A  \otimes G_{B>C})\phi^-\|
 \le \| (G_{ABC} - |0\rangle\langle 0|_A \otimes G_{BC}) |v\rangle\langle v|_C^\perp \|
 \le O(j^{1/2}|\lambda|^{-j}).
\end{equation}
Here the last inequality follows from  Lemma~\ref{lemma:C2}, where we set $D=\emptyset$.

Next  let us consider $\phi^{-+}$. Using the definitions of $G_{AB>C}$ and $G_{B>C}$ one gets
\[
G_{AB>C} \cdot \phi^{-+}=(G_{AB} - G_{ABC}) \cdot \phi^{-+}
\quad \mbox{and} \quad
|0\rangle\langle0|_A \otimes G_{B>C} \cdot \phi^{-+} = |0\rangle\langle0|_A\otimes (G_B-G_{BC})\cdot \phi^{-+}.
\]
It follows that 
\begin{eqnarray}
 \| (G_{AB>C}  - |0\rangle\langle 0|_A  \otimes G_{B>C})\phi^{-+}\| &\le &
 \| (G_{AB'B''} - |0\rangle\langle0|_A \otimes G_{B'B''}) |v\rangle\langle v|_{B''}^\perp \| \nonumber \\
 && + \| (G_{AB'B''C}- |0\rangle\langle0|_A \otimes G_{B'B''C}) |v\rangle\langle v|_{B''}^\perp \|
\end{eqnarray}
We can bound both terms in the right-hand side of the above equations
using Lemma~\ref{lemma:C2}.
One should 
 choose the four regions 
in the statement of Lemma~\ref{lemma:C2} as $(A,B,C,D)=(A,B',B'',\emptyset)$ for the first term
and $(A,B,C,D)=(A,B',B'',C)$ for the second term. This results in
\begin{equation}
\label{Xstep3}
\| (G_{AB>C}  - |0\rangle\langle 0|_A  \otimes G_{B>C})\phi^{-+}\| 
\le O(j^{1/2} |\lambda|^{-j/2}) +  O(j^{1/2} |\lambda|^{-j/2})= O(j^{1/2} |\lambda|^{-j/2}).
\end{equation}

Finally, let us consider $\phi^{++}$. We have
\[
G_{AB>C} \cdot \phi^{++} = (G_{AB} - G_{ABC})\cdot \phi^{++}
\quad \mbox{and} \quad |0\rangle\langle 0|_A\otimes G_{B>C} \cdot \phi^{++}
=|0\rangle\langle 0|_A \otimes (G_B- G_{BC})\cdot \phi^{++}.
\]
It follows that 
\begin{equation}
\label{Xstep4}
\| G_{AB>C} \cdot \phi^{++}\| \le \| (G_{(AB')B''C}- G_{(AB')B''} ) |v\rangle\langle v|_{B''C} \| \le O\left(j^{1/2} |\lambda|^{-j/4}\right)
+ O\left(|c|^{j/4}\right).
\end{equation}
Here we applied  Lemma~\ref{lemma:C1} choosing  the regions in the statement of the lemma
 as $(A,B,C)=(AB',B'',C)$.
Likewise,
\begin{eqnarray}
\| |0\rangle\langle 0|_A\otimes G_{B>C} \cdot  \phi^{++} \|
&\le &  \| (G_{BC}- G_B)\cdot \phi^{++}\| \le  \| (G_{B'B''C}- G_{B'B''} ) |v\rangle\langle v|_{B''C} \| \nonumber \\
&\le & 
O\left(j^{1/2} |\lambda|^{-j/4}\right)
+ O\left(|c|^{j/4}\right). \label{Xstep5}
\end{eqnarray}
Here we applied  Lemma~\ref{lemma:C1} choosing the regions in 
the statement of the lemma
as $(A,B,C)=(B',B'',C)$.
Substituting Eqs.~(\ref{Xstep2},\ref{Xstep3},\ref{Xstep5}) into Eq.~(\ref{Xstep1})
and using the triangle inequality one arrives at the desired bound Eq.~(\ref{X1}).
\end{proof}

\subsection{Proof of the gapped phase theorem} 
\label{subs:finalgappedproof}
Let us now prove Theorem~\ref{thm:Nbound} for the case where $\psi$ is entangled, that is, $\mu_1\neq 0$. In Section~\ref{subs:simple} we will see how to modify this proof to handle the product state case $\mu_1=0$.

Partition region $B$ as $B=B'B''$, where $|B'|=\lfloor r/2\rfloor, |B''|=\lceil r/2\rceil$. 
First we note that 
\begin{equation}
\label{Tstep1}
G_{AB}G_{BC}-G_{ABC}=(G_{AB}-G_{ABC})G_{BC}
\end{equation}
and
\[
G_{AB}-G_{ABC}=G_{AB} \otimes |v\rangle\langle v|_C + G_{AB} \otimes |v\rangle\langle v|_C^\perp-
G_{ABC}=
G_{AB>C} + G_{AB} \otimes |v\rangle\langle v|_C^\perp.
\]
Here we used the notation from Lemma~\ref{lemma:X}.
Applying Lemma \ref{lemma:X} to exclude region $A$ from $G_{AB>C}$ one gets
\begin{equation}
\label{Tstep2}
G_{AB} -G_{ABC}= |0\rangle\langle 0|_A \otimes G_{B>C} + 
G_{AB} \otimes |v\rangle\langle v|_C^\perp +\epsilon_r
\end{equation}
where $\epsilon_r$ denotes some operator such that 
\[
\| \epsilon_r\| \le O\left(r^{1/2} |\lambda|^{-r/4}\right) + 
O\left( |c|^{r/4}\right).
\]
Substituting the identity
\[
G_{B>C}=G_B\otimes |v\rangle\langle v|_C - G_{BC} = (G_B-G_{BC}) - G_B \otimes |v\rangle\langle v|_C^\perp
\]
 one gets
\begin{align}
G_{AB}-G_{ABC}&=|0\rangle\langle 0|_A \otimes(G_B-G_{BC})+(G_{AB}-|0\rangle\langle 0|_A \otimes G_B ) \otimes |v\rangle\langle v|_C^\perp 
+\epsilon_r \nonumber\\ 
&=|0\rangle\langle 0|_A \otimes(G_B-G_{BC})-G_{A<B'B''} \otimes  |v\rangle\langle v|_C^\perp  + \epsilon_r. \label{Tstep3}
\end{align}
Applying Lemma \ref{lemma:X} to exclude region $B''$ from $G_{A<B'B''}$
one gets
\begin{equation}
\label{Tstep4}
G_{AB}-G_{ABC}=|0\rangle\langle 0|_A \otimes(G_B-G_{BC})-G_{A<B'} \otimes |v\rangle\langle v|_{B''} \otimes  |v\rangle\langle v|_C^\perp
+ \epsilon_{r/2} + \epsilon_r.
\end{equation} 
Using $(G_B-G_{BC})G_{BC}=0$, $\|G_{A<B'}\|\le 2$, and $G_{BC}=G_{B''C} G_{BC}$, one arrives at
\begin{equation}
\label{Tstep5}
\|(G_{AB}-G_{ABC})G_{BC}\| \le 2\| (|v\rangle\langle v|_{B''} \otimes  |v\rangle\langle v|_C^\perp)G_{B''C}\|
+\|\epsilon_{r/2}\| + \|\epsilon_r\|.
\end{equation}
Finally, partition $B''=B''_1B''_2$  with $|B''_1|=\lfloor \; |B''|\; \rfloor$ and $|B''_2|=\lceil \; |B''|\; \rceil$ (so that each part has size $\approx r/4$).
Denote
\[
\delta_{r/4}=|v\rangle\langle v|_C^\perp (G_{B''C}-|0\rangle\langle0|_{B''_1} \otimes G_{B''_2C} ).
\]
Applying  Lemma~\ref{lemma:C2} where the four regions are chosen as
$(A,B,C,D)=(B''_1,B''_2,C,\emptyset)$ one concludes that 
\begin{equation}
\label{Tstep6}
\|\delta_{r/4}\| \le O\left(r^{1/2} |\lambda|^{-r/4}\right).
\end{equation}
Replacing $G_{B''C}$ in Eq.~(\ref{Tstep5}) by $|0\rangle\langle0|_{B''_1} \otimes G_{B''_2C}$
and using Eq.~(\ref{Tstep6})
results in 
\begin{align}
\|(G_{AB}-G_{ABC})G_{BC}\| &\le   2\| \, |v\rangle\langle v|_{B''} \cdot |0\rangle\langle0|_{B''_1}\,  \|
+ 2\|\delta_{r/4}\| + \|\epsilon_{r/2}\| + \|\epsilon_r\| \nonumber \\
&\le  2| \langle v |0\rangle|^{|B''_1|} + O\left( r^{1/2} |\lambda|^{-r/8}\right) + O\left(|c|^{r/8}\right) \nonumber\\
&\le  O\left( r^{1/2} |\lambda|^{-r/8}\right) + O\left(|c|^{r/8}\right).
\end{align}
This completes the proof of Theorem~\ref{thm:Nbound}
for the case when $\psi$ is an entangled state.

\subsection{Specializing to product states}
\label{subs:simple}
Finally consider the case $\mu_1=0$. This implies that  
$\det(T_\psi)=0$, that is, $\psi$ is a product state. Using the notation from Section~\ref{subs:open} write
\[
|\psi\rangle=|1v^\perp\rangle, \quad |v\rangle=c|0\rangle+s|1\rangle, \quad |v^\perp\rangle=s|0\rangle-c^*|1\rangle,
\]
where $|c|^2+s^2=1$. Here $s\neq 0$ which follows from the fact that $\mu_2\neq 0$. It is also easily checked that $|0\rangle$ and $|v\rangle$ are eigenvectors of $T_\psi$ corresponding to eigenvalues $\mu_1=0,\mu_2=-s$ respectively.

We now show that the region exclusion identities presented in Lemmas~\ref{lemma:C1},\ref{lemma:C2},\ref{lemma:X} (for the case of entangled $\psi$) become exact equalities.
Indeed, as was shown in Section~\ref{subs:open},
the ground space of $H_n(\psi)$  has an orthonormal basis
$g_0,\ldots,g_n$, where 
$|g_0\rangle=|v^{\otimes n}\rangle$ and $|g_i\rangle=|0^{i-1} v^\perp v^{n-i}\rangle$
for $i\ge 1$.
Thus 
\begin{equation}
\label{Gexplicit}
G_n=\sum_{i=0}^n |g_i\rangle\langle g_i|,
\end{equation}

We start with Lemma~\ref{lemma:C2}. Let $[n]=ABCD$
be an arbitrary    partition such that $B$ and $C$ are non-empty. 
We have to prove that 
\begin{equation}
\label{C2exact}
(G_{ABCD}- |0\rangle\langle 0|_A \otimes G_{BCD})\cdot  |v\rangle\langle v|_C^\perp =0.
\end{equation}
Note that $|v\rangle\langle v|_j\cdot |g_i\rangle=|g_i\rangle$ for all $i<j\le n$. 
This implies 
\[
|g_i\rangle\langle g_i| \cdot |v\rangle\langle v|_C^\perp = |g_i\rangle\langle g_i|
\cdot |v\rangle\langle v|_C \cdot |v\rangle\langle v|_C^\perp =0 \quad \mbox{for all $i\in AB$}.
\]
Substituting Eq.~(\ref{Gexplicit}) for $G_{ABCD}$ and using the above identity yields 
\[
G_{ABCD} \cdot |v\rangle\langle v|_C^\perp =\sum_{i\in CD} |g_i\rangle\langle g_i| \cdot
|v\rangle\langle v|_C^\perp =
|0\rangle\langle 0|_A \otimes G_{BCD} \cdot |v\rangle\langle v|_C^\perp 
\]
since $|g_i\rangle=|0\rangle_A \otimes |g_{i-|A|}\rangle_{BCD}$ for all $i\in CD$.
This is equivalent to Eq.~(\ref{C2exact}).

Next consider Lemma~\ref{lemma:C1}. 
Let $[n]=ABC$ be any partition such that 
$B$ is non-empty.  We have to prove that 
\begin{equation}
\label{C1exact}
(G_{ABC}-G_{AB}\otimes I_C) \cdot |v\rangle\langle v|_{BC} =0.
\end{equation}
Note that for any $i\in A$ one has
\[
|g_i\rangle\langle g_i|=|g_i\rangle\langle g_i|_A \otimes |v\rangle\langle v|_{BC}
\quad \mbox{and} \quad
|g_i\rangle\langle g_i|_{AB}=|g_i\rangle\langle g_i|_A \otimes |v\rangle\langle v|_B.
\]
On the other hand, $|g_i\rangle\langle g_i|\cdot |v\rangle\langle v|_{BC} =0$
for any $i\in BC$. 
Therefore
\[
G_{ABC}\cdot |v\rangle\langle v|_{BC} =\sum_{i\in A\cup \{0\}} |g_i\rangle\langle g_i|\cdot |v\rangle\langle v|_{BC} 
= \sum_{i\in A\cup \{0\}} |g_i\rangle\langle g_i|_A \otimes |v\rangle\langle v|_{BC} 
\]
Likewise, $|g_i\rangle\langle g_i|_{AB} \cdot |v\rangle\langle v|_{BC}=0$
for any $i\in B$. Therefore
\[
(G_{AB}\otimes I_C)\cdot |v\rangle\langle v|_{BC} = \sum_{i\in A\cup \{0\}} |g_i\rangle\langle g_i|_{AB}
\cdot  |v\rangle\langle v|_{BC} = \sum_{i\in A\cup \{0\}} |g_i\rangle\langle g_i|_A \otimes  |v\rangle\langle v|_{BC}.
\]
Comparing the last two identities one arrives at Eq.~(\ref{C1exact}).

Finally, note that the proof of Lemma~\ref{lemma:X} only uses Lemmas~\ref{lemma:C1},\ref{lemma:C2}, and the fact that Eq.~\eqref{X2} is equivalent to  Eq.~\eqref{X1} applied to the left-right flipped chain. In the proof of Lemma~\ref{lemma:X} we use the fact that $T_\psi$ is invertible to establish this latter fact. In the case at hand $T_\psi$ is not invertible but since $|\psi\rangle=|1v^\perp\rangle$ we immediately see that  Eq.~\eqref{X2} is just the left-right flipped version of  Eq.~\eqref{X1}. So both inequalities in Lemma~\ref{lemma:X} become exact equalities.

The proof of Theorem~\ref{thm:Nbound} from Section~\ref{subs:finalgappedproof} uses the Region Exclusion lemmas to establish the result. Since we have shown that each of these lemmas holds (with exact equality) for product states $\psi$, we see that the proof of Theorem~\ref{thm:Nbound} also applies in this case if one formally sets $\lambda=\infty$ in all error terms. 


\section{Acknowledgments}
DG thanks Edward Farhi, Jeffrey Goldstone, and Sam Gutmann for many discussions about this problem when he was a graduate student.  DG was supported in part by NSERC and by ARO. IQC is supported in part by the Government of Canada and the Province of Ontario.
SB  would like to thank Ramis Movassagh for helpful discussions. 
SB acknowledges NSF Grant CCF-1110941.

\setcounter{secnumdepth}{0}
\appendix
\section{Appendix}
\subsection{Qubit chains with higher rank projectors} 
\label{appendix}
In the main body of the paper we considered qubit chains where $\Pi$ is rank-$1$ and $H_n(\Pi)$, defined in Eq.~\eqref{Hnpi}, is guaranteed to be frustration-free. Here we consider qubit chains where $\Pi$ is rank-$2$ or rank-$3$ and we determine which projectors $\Pi$ correspond to frustration-free chains. For each frustration-free chain we determine if the system is gapped or gapless. We say that $H_n(\Pi)$ is gapped if its spectral gap, denoted $\gamma(\Pi,n)$, is lower bounded by a positive constant independent of $n$ (otherwise it is gapless). We shall write $\mathcal{G}_n$ for the null space of $H_n(\Pi)$ and $G_n$ for the projector onto this space.

The case where $\Pi$ is rank-$3$ is trivial, so we consider it first. In this case there is a unique two-qubit state $|\chi\rangle$ in the null space of $\Pi$.  If $H_n(\Pi)$ is frustration-free then there exists an $n$-qubit state $\psi$ with reduced state on each pair of consecutive qubits $i,i+1$  supported entirely on $\chi$. It follows that the rank-$3$ chain is frustration-free if and only if  $\chi=\theta\otimes \theta$ for some single-qubit state $\theta$. Thus $\Pi=I-|\theta\rangle\langle\theta|^{\otimes 2}$ and $H_n(\Pi)$ is a sum of pairwise commuting projectors.
This shows that  $H_n(\Pi)$ has unique ground state $|\theta\rangle^{\otimes n}$, and its eigenvalue gap is equal to $1$, for all $n\geq 2$.

The rank-$2$ case is slightly more interesting. There is a trivial case where $\Pi$ is a $1$-local projector, i.e., $\Pi=P\otimes I$ or $\Pi=I\otimes P$; in that case $\gamma(\Pi,n)=1$ for all $n\geq 2$. The following theorem handles all other cases.

\begin{theorem}
\label{thm:rank2}
Suppose $\Pi$ is a two-qubit, rank-2 projector which cannot be written as $I\otimes P$ or $P\otimes I$ for some projector $P$. Then the dimension of the null space of $H_n(\Pi)$ satisfies $\mathrm{dim}(\mathcal{G}_n)=\mathrm{dim}(\mathcal{G}_4)\in \{0,1,2\}$ for all $n\geq 4$. Moreover, exactly one of the following holds:
\begin{enumerate}
\item $\mathcal{G}_4=\mathrm{span} \{|\alpha \alpha \alpha \alpha \rangle\}$ for some single-qubit state $\alpha$.
\item $\mathcal{G}_4=\mathrm{span} \{|\alpha \alpha \alpha \alpha \rangle, |\beta  \beta \beta \beta \rangle\}$ for some linearly independent single-qubit states $\alpha,\beta$.
\item $\mathcal{G}_4=\mathrm{span} \{|\alpha \beta \alpha \beta \rangle, |\beta  \alpha \beta \alpha \rangle\}$ for some linearly independent single-qubit states $\alpha,\beta$.
\item $\mathcal{G}_4=\mathrm{span} \{|\alpha \alpha \alpha \alpha \rangle, |\alpha^\perp  \alpha \alpha \alpha \rangle+f|\alpha \alpha^\perp \alpha \alpha \rangle+f^2|\alpha \alpha \alpha^\perp \alpha \rangle+f^3|\alpha \alpha \alpha \alpha^\perp \rangle\}$ for some orthonormal single-qubit states $\alpha,\alpha^\perp$ and non-zero $f\in \CC$.
\item $\mathcal{G}_4$ is empty.
\end{enumerate}
In cases 1, 2, 3, and 4, the Hamiltonian $H_n(\Pi)$ is frustration-free for all $n\geq 2$, whereas in case 5 it is frustrated for $n\geq 4$. $H_n(\Pi)$ is gapped in cases 1, 2, and 3, and it is gapped in case 4 if $|f|\neq 1$. If $|f|=1$ in case 4 then the Hamiltonian is gapless, with spectral gap upper bounded as $\gamma(\Pi,n)\leq (1-\cos(\pi/n))$. 
\end{theorem}
To establish lower bounds on spectral gaps we shall use Nachtergaele's criterion~\cite{Nachtergaele1996}. Recall that Lemma \ref{lemma:N} states this criterion for the case where $\Pi$ is rank-$1$. More generally if $H_n(\Pi)$ is frustration-free for $n\geq 2$ then the same bound holds for its spectral gap (with $\gamma(\psi,n)$  and $\gamma(\psi,r+1)$ replaced by $\gamma(\Pi,n)$ and $\gamma(\Pi,r+1)$ in the statement of the Lemma). For our purposes it will be sufficient to use the following weaker version of the bound.
\begin{Lemma}
\label{lem:gapcor}
Suppose $H_n(\Pi)$ is frustration-free for $n\geq 2$. Let  $[n]=ABC$ with $|C|=1$, $|B|=r$, and $|A|=n-r-1$. Suppose there exist constants $0\leq \delta<1$ and $K>0$ such that for all sufficiently large $n$ we have $\|G_{ABC}-G_{AB}G_{BC}\|\leq K\delta^{r}$. Then $H_n(\Pi)$ is gapped.
\end{Lemma}
\begin{proof}
We can always choose $r$ so that $K \delta^{r}\leq \frac{1}{2\sqrt{r+1}}$. Plugging this choice into Nachtergaele's bound with $\epsilon=\frac{1}{2\sqrt{r+1}}$ we obtain $\gamma(\Pi,n)\geq \gamma(\Pi,r+1)(4r+4)^{-1}$.
\end{proof}

\begin{proof}[\bf Proof of Theorem~\ref{thm:rank2}]
We first establish that the range of $\Pi$ is spanned by two linearly independent states $\phi,\psi$ which are both entangled. It is easy to check that the only two dimensional subspaces of $\CC^{2} \otimes \CC^{2}$ which contain only product states are of the form $\chi\otimes \CC^2$ or $\CC^2 \otimes \chi$ for some single qubit state $\chi$. By the hypothesis of the theorem $\Pi$ cannot be written as $I\otimes P$ or $P\otimes I$, which implies $\mathrm{range}(\Pi)$ does not have this form; thus it contains at least one entangled state $\phi$. Let $\nu\in \mathrm{range}(\Pi)$ be linearly independent from $\phi$. It is easy to see that we can always choose $\psi=\phi+z\nu$ for some nonzero $z\in \CC$ so that $\psi$ is entangled.

So $\mathrm{range}(\Pi)=\mathrm{span}\{\phi,\psi\}$ where $\psi,\phi$ are both entangled (equivalently, $T_\phi$ and $T_\psi$ are both invertible). Furthermore, an $n$-qubit state $\chi$ is in the null space of $H_n(\Pi)$ if and only if it is in the null space of $|\psi\rangle\langle \psi|_{i,i+1}$ and $|\phi\rangle\langle \phi|_{i,i+1}$ for all $i=1,\ldots,n-1$.

To complete the proof we now suppose that $\mathcal{G}_4$ is nonempty and we consider two cases depending on whether or not $T_\phi^{-1}T_\psi$ has two linearly independent eigenvectors. The theorem follows directly from the following propositions which handle the two cases.
\begin{prop}
Suppose $\Pi$ is a two-qubit projector such that  $\mathrm{range}(\Pi)=\mathrm{span}\{\phi,\psi\}$ where $\psi,\phi$ are both entangled. Suppose that $T_\phi^{-1}T_\psi$ has linearly independent eigenvectors $\{\alpha,\beta\}$ and that $\mathcal{G}_4$ is nonempty.  Then one of the cases 1,2, or 3 from Theorem \ref{thm:rank2} occurs. Moreover, $H_n(\Pi)$ is frustration-free for all $n\geq 2$ and it is gapped. Its ground space dimension satisfies $\mathrm{dim}(\mathcal{G}_n)=\mathrm{dim}(\mathcal{G}_4)$ for all $n\geq 4$.
\label{prop:Teig1}
\end{prop}

\begin{prop}
Suppose $\Pi$ is a two-qubit projector such that  $\mathrm{range}(\Pi)=\mathrm{span}\{\phi,\psi\}$ where $\psi,\phi$ are both entangled. Suppose that $T_\phi^{-1}T_\psi$ has only one linearly independent eigenvector $\{\alpha\}$ and that $\mathcal{G}_4$ is nonempty.  Then case 4 from Theorem \ref{thm:rank2} occurs. Moreover, $H_n(\Pi)$ is frustration-free and has ground space dimension equal to $2$ for all $n\geq 2$. It is gapped if and only if $|f|\neq 1$; if $|f|=1$ then $\gamma(\Pi,n)\leq (1-\cos(\pi/n))$.
\label{prop:Teig2}
\end{prop}
\end{proof}
In the remainder of this section we prove Propositions \ref{prop:Teig1} and \ref{prop:Teig2}.

\vspace{.25cm}
\noindent
\textbf{Proof of Proposition \ref{prop:Teig1}}
We first establish that $\mathcal{G}_2$ is spanned by $1\otimes T_\psi |\alpha\rangle|\alpha\rangle$ and $1\otimes T_\psi |\beta\rangle|\beta\rangle$. These states are linearly independent (since $\alpha,\beta$ are). Since $\mathcal{G}_2$ is $2$-dimensional it suffices to establish that it contains both of these states. Clearly each is orthogonal to $|\psi\rangle \langle \psi|$ since $|\psi\rangle\sim \langle \epsilon|I\otimes T_\psi^{-1}$ (recall $|\epsilon\rangle=|0,1\rangle-|1,0\rangle$). To check that they are orthogonal to $|\phi\rangle\langle\phi|$ use the fact that $\langle \phi| \sim \langle \epsilon|I\otimes T_\phi^{-1}$ and that $\alpha,\beta$ are eigenvectors of $T_\phi^{-1}T_\psi$.

Now consider $n=3$.  Since $\mathcal{G}_4$ is nonempty, there exists a state $|\chi\rangle \in \mathcal{G}_3$, which by Proposition \ref{gsdegen} can be written as $|\chi\rangle=1\otimes T_\psi\otimes T_\psi^2 |s\rangle$ for some state $|s\rangle$ in the $3$-qubit symmetric subspace. Since the chain is frustration free, the first two qubits of $\chi$ have all of their support in $\mathcal{G}_2$. This implies $|s\rangle=a|\alpha\rangle^{\otimes 3}+b|\beta\rangle^{\otimes 3}$ where $a,b$ are not both zero. By symmetry we assume without loss of generality that $a\neq 0$. Next, imposing orthogonality to $|\phi\rangle\langle \phi|_{2,3}$ gives
\[
a |\alpha\rangle \left(\langle \epsilon|T_\psi \otimes T_\phi^{-1}T_\psi^2|\alpha,\alpha\rangle\right)+b|\beta\rangle \left(\langle \epsilon|T_\psi \otimes T_\phi^{-1}T_\psi^2|\beta,\beta\rangle\right)=0.
\]
Since $\alpha,\beta$ are linearly independent, both terms must be zero. Using the fact that any two-qubit state orthogonal to $\epsilon$ is symmetric, and that $a\neq 0$, we get:
\begin{itemize}
\item $T_\psi |\alpha\rangle$ is an eigenvector of $T_\phi^{-1}T_\psi$, and
\item If $b\neq 0$ then $T_\psi |\beta\rangle$ is an eigenvector of $T_\phi^{-1}T_\psi$.
\end{itemize}
Now recall that $\alpha,\beta$ are linearly independent eigenvectors of $T_\phi^{-1}T_\psi$. Note that $T_\phi^{-1}T_\psi$ is not proportional to the identity (since this would imply that $\phi$ is proportional to $\psi$) and therefore $\alpha,\beta$  are the only eigenvectors of $T_\phi^{-1}T_\psi$. Hence $T_\psi |\alpha\rangle$ is proportional to one of the states $\alpha, \beta$, and if $b\neq 0$ then the same holds for $T_\psi|\beta\rangle$.  However, since $T_\psi$ is invertible, it cannot be the case that $T_\psi |\alpha\rangle \sim T_\psi |\beta\rangle$. Putting this together we see there are 4 subcases to consider (below we show that the first three correspond to cases 1., 2., 3., from the statement of the theorem and that the fourth does not occur):
\medskip
\newline
\noindent
\textbf{Case (i): $T_\psi |\alpha\rangle\sim |\alpha\rangle$ and $T_\psi |\beta\rangle =c|\alpha\rangle +d|\beta\rangle$ where $c,d$ are both non-zero.}\newline
In this case $T_\psi |\beta\rangle$ is not an eigenvector of $T_\phi^{-1}T_\psi$, which implies (by the facts established above) that $b=0$ and thus $|\chi\rangle\sim |\alpha\rangle^{\otimes 3}$ is the only state in $\mathcal{G}_3$. Likewise, $|\alpha\rangle^{\otimes n}$ is the only state in $\mathcal{G}_n$ for $n\geq 3$ (since it is the unique $n$-qubit state such that any three consecutive qubits $i,i+1,i+2$ have all of their support on $\mathcal{G}_3$.) This establishes that we are in case 1 from Theorem \ref{thm:rank2}. Note that for any partition $[n]=ABC$ with $|B|\geq 1$ we have $G_{ABC}=|\alpha\rangle\langle \alpha|^{\otimes n}=G_{AB}G_{BC}$ and so the conditions of Lemma \ref{lem:gapcor} are satisfied (with $\delta=0$) and $H_n(\Pi)$ is gapped.
\medskip
\newline
\noindent
\textbf{Case (ii): $T_\psi |\alpha\rangle\sim |\alpha\rangle$ and $T_\psi |\beta\rangle\sim |\beta\rangle$ }
\newline
In this case (by the facts established above) $\mathcal{G}_2$ is spanned by $|\alpha\rangle|\alpha\rangle$ and $|\beta\rangle|\beta\rangle$. This implies that $\mathcal{G}_n$ is spanned by $|\alpha\rangle^{\otimes n}$ and $|\beta\rangle^{\otimes n}$ for all $n\geq 2$, so we are in case 2 of Theorem \ref{thm:rank2}. One can easily construct an orthonormal basis and confirm that $\|G_n-|\alpha\rangle \langle \alpha|^{\otimes n}-|\beta\rangle \langle \beta|^{\otimes n}\|=O(|\langle \alpha|\beta\rangle|^n)$. Using this expression three times and the triangle inequality we get $\|G_{AB}G_{BC}-G_{ABC}\|\leq K|\langle \alpha|\beta\rangle|^{|B|}$ where $K$ is a constant, for any partition $[n]=ABC$. Hence the conditions of Lemma \ref{lem:gapcor} are satisfied (with $\delta=|\langle \alpha|\beta\rangle|$) and $H_n(\Pi)$ is gapped.
\medskip
\newline
\noindent
\textbf{Case (iii): $T_\psi |\alpha\rangle\sim |\beta\rangle$ and $T_\psi |\beta\rangle\sim |\alpha\rangle$}\newline
In this case $\mathcal{G}_2=\mathrm{span}\{|\alpha\rangle |\beta\rangle,|\beta\rangle |\alpha\rangle\}$, which implies that for all $n\geq 2$, $\mathcal{G}_n$ is spanned by product states $|\nu_1\rangle=|\alpha\rangle |\beta\rangle |\alpha\rangle\ldots |\beta\rangle$ and  $|\nu_2\rangle=|\beta\rangle |\alpha\rangle |\beta\rangle\ldots |\alpha\rangle$ (the last tensor product factors are instead $|\alpha\rangle, |\beta\rangle$ respectively if $n$ is odd). This shows that we are in case 3 of Theorem \ref{thm:rank2}. Now constructing an orthonormal basis we see that $\|G_n-|\nu_1\rangle\langle \nu_1|-|\nu_2\rangle\langle \nu_2|\|=O(|\langle \alpha|\beta\rangle|^{n}).$ Letting $[n]=ABC$ and using this expression we get $\|G_{AB}G_{BC}-G_{ABC}\|\leq K|\langle \alpha|\beta\rangle|^{|B|}$ for some constant $K$. Hence the conditions of Lemma \ref{lem:gapcor} are satisfied (with $\delta=|\langle \alpha|\beta\rangle|$) and $H_n(\Pi)$ is gapped.
\medskip
\newline
\noindent
\textbf{Case (iv): $T_\psi |\alpha\rangle\sim |\beta\rangle$ and $T_\psi |\beta\rangle =c|\alpha\rangle +d|\beta\rangle$ where $c,d$ are both non-zero.}\newline
In this final case $b=0$ and the only state in $\mathcal{G}_3$ is $|\chi\rangle=1\otimes T_\psi\otimes T_\psi^2|\alpha\rangle |\alpha\rangle| \alpha\rangle \sim |\alpha\rangle|\beta\rangle T_\psi |\beta\rangle$. Any non-zero state $\kappa \in \mathcal{G}_4$ must have its first three and last three qubits in the state $\chi$. This implies in particular that $\kappa$ is a product state. Furthermore, the second qubit of $\kappa$ must be in the state $\beta$ (look at the first three qubits) and also in the state $\alpha$ (look at the last three qubits), which is impossible. Therefore $\mathcal{G}_4$ is empty, which is a contradiction. So case (iv) does not occur.
\qed

\vspace{0.25cm}
\textbf{Proof of Proposition \ref{prop:Teig2}}
Define an orthonormal basis $\{|\hat{0}\rangle=|\alpha\rangle, |\hat{1}\rangle=|\alpha^\perp\rangle\}$. Since $|\hat{0}\rangle$ is the only eigenvector of $T_\phi^{-1}T_\psi$, in this basis we have
\begin{equation}
T_\phi^{-1} T_\psi=\left(\begin{array}{rr}
c & d\\
0& c
\end{array}\right)
\label{Tphipsi}
\end{equation}
for some $c,d\in \mathbb{C}$ with $c\neq 0$ (since $\det\left(T_\phi^{-1}T_\psi\right)\neq 0$).  

First consider $n=2$. We claim that $1\otimes T_\psi|\hat{0}\rangle|\hat{0}\rangle$ and $1\otimes T_\psi\left(|\hat{0}\rangle|\hat{1}\rangle+|\hat{1}\rangle|\hat{0}\rangle\right)$ span $\mathcal{G}_2$. To see this note that these states are linearly independent and orthogonal to $|\psi\rangle\langle \psi|$. To show that they are also orthogonal to $|\phi\rangle\langle \phi|$, use the fact that $\langle \phi| \sim \langle \epsilon|1\otimes T_\phi^{-1}$ and Eq.~\eqref{Tphipsi}.

We now show that our assumption that $\mathcal{G}_4$ (and therefore also $\mathcal{G}_3$) is nonempty implies $1\otimes T_\psi \otimes T_\psi^2 |\hat{0},\hat{0},\hat{0}\rangle$ is in $\mathcal{G}_3$. To reach a contradiction, assume $1\otimes T_\psi \otimes T_\psi^2 |\hat{0},\hat{0},\hat{0}\rangle\notin \mathcal{G}_3$. By Proposition \ref{gsdegen} any state $\omega\in \mathcal{G}_3$ satisfies $|\omega\rangle =I\otimes T_\psi\otimes T_\psi^2|s\rangle$ for some $|s\rangle$ in the three-qubit symmetric subspace. The first two qubits of $|s\rangle$ must be supported entirely in $\mathcal{G}_2$. Using the form of $\mathcal{G}_2$ derived above, we see that this implies $|s\rangle$ is a superposition of the symmetric Hamming weight zero and one states (with respect to the $\hat{0},\hat{1}$ basis). Since $1\otimes T_\psi \otimes T_\psi^2 |\hat{0},\hat{0},\hat{0}\rangle\notin\mathcal{G}_3$ (and since $\mathcal{G}_3$ is nonempty), we have shown that $\mathcal{G}_3$ is one-dimensional and contains a single state $\omega$ of the form
\[
|\omega\rangle=1\otimes T_\psi \otimes T_\psi^2 \left(a|\hat{0},\hat{0},\hat{0}\rangle +b|\hat{1},\hat{0},\hat{0}\rangle+b|\hat{0},\hat{1},\hat{0}\rangle+b|\hat{0},\hat{0},\hat{1}\rangle\right)
\]
where $b\neq 0$. By assumption there exists a state in $\mathcal{G}_4$, which must have its last three and first three qubits each supported on $\mathcal{G}_3$, i.e., it must be of the form $|\theta_1\rangle|\omega\rangle=|\omega\rangle|\theta_2\rangle$ for some single-qubit states $\theta_1,\theta_2$. However this is impossible since $\omega$ is not a product state whenever $b\neq 0$. Having reached a contradiction we conclude $1\otimes T_\psi \otimes T_\psi^2 |\hat{0},\hat{0},\hat{0}\rangle$ is in $\mathcal{G}_3$.

Now 
\[
_{2,3}\langle \phi|1\otimes T_\psi \otimes T_\psi^2 |\hat{0},\hat{0},\hat{0}\rangle \sim \; |\hat{0}\rangle\left(\langle \epsilon|T_\psi \otimes T_\phi^{-1}T_\psi^2 |\hat{0},\hat{0}\rangle\right)=0
\]
implies that $T_\psi|\hat{0}\rangle$ is an eigenvector of $T_\phi^{-1}T_\psi$. By assumption this operator has only one eigenvector $\hat{0}$, so $T_\psi|\hat{0}\rangle\sim |\hat{0}\rangle$. Thus $|\hat{0}\rangle$ is an eigenvector of both $T_\psi$ and $T_\phi$. Therefore, in the $\hat{0}, \hat{1}$ basis we have 
\[
T_\psi=\left(\begin{array}{rr}
x & y\\
0& z
\end{array}\right) \qquad 
T_\phi=\left(\begin{array}{rr}
q & r\\
0& s
\end{array}\right)
\qquad
T_\phi^{-1}T_\psi \sim\left(\begin{array}{rr}
sx & sy-rz\\
0& qz
\end{array}\right)
\]
for some $x,y,z,q,r,s\in \CC$. Comparing with Eq.~\eqref{Tphipsi} we see that $sx=qz\neq 0$. Now 
\begin{equation}
|\psi\rangle=x^* |\hat{0},\hat{1}\rangle-z^*|\hat{1},\hat{0}\rangle+y^*|\hat{1},\hat{1}\rangle 
\label{psixyz}
\end{equation}
and
\begin{equation}
|\phi\rangle=q^* |\hat{0},\hat{1}\rangle-s^*|\hat{1},\hat{0}\rangle+r^*|\hat{1},\hat{1}\rangle= (q^*/x^*)\left(x^* |\hat{0},\hat{1}\rangle-z^*|\hat{1},\hat{0}\rangle\right)+r^*|\hat{1},\hat{1}\rangle.
\label{phiqrs}
\end{equation}
Comparing Eqs.~\eqref{psixyz},\eqref{phiqrs} we see that $\mathrm{range}(\Pi)=\mathrm{span}\{|\phi\rangle,|\psi\rangle\}=\mathrm{span}\{|\hat{1},\hat{1}\rangle,|\nu\rangle\}$ where
\begin{equation}
|\nu\rangle=\frac{1}{\sqrt{1+|f|^2}}\left(|\hat{0},\hat{1}\rangle-f|\hat{1},\hat{0}\rangle\right)
\label{eq:Pi_orth2}
\end{equation}
and $f=z^*/x^*$ (note that $x,z$ are both nonzero since  $sx=qz\neq 0$). Moreover these states are orthonormal so $\Pi=|\hat{1},\hat{1}\rangle\langle\hat{1},\hat{1}|+|\nu\rangle\langle \nu|$.

By Proposition \ref{gsdegen}, a basis for the zero energy ground space of  $\sum_i |\nu\rangle\langle\nu|_{i,i+1}$ is given by
\[
T_{\nu}^{\mathrm{all}}\sum_{z\in\{0,1\}^{\otimes n}, |z|=j}|\hat{z}\rangle, \qquad  j=0\ldots n
\]
where $T_\nu\sim \mathrm{diag}(1,f)$ and $T_{\nu}^{\mathrm{all}}=1\otimes T_\nu\otimes T_\nu^2\otimes\ldots \otimes T_\nu^{n-1}$. It is then easy to see that the only states in this space orthogonal to $\sum_i |\hat{1},\hat{1}\rangle\langle\hat{1},\hat{1}|_{i,i+1}$ are the basis vectors corresponding to $j=0,1$ and linear combinations thereof. These two states span $\mathcal{G}_n$ for all $n\geq 2$; this shows that we are in case 4 of Theorem \ref{thm:rank2}. 

We now show that the system is gapped if $|f|\neq1$. Without loss of generality we may assume $|f|<1$; note that if $|f|>1$ we may relabel the qubits $1,2,3,\ldots,n$ as $n,n-1,\ldots,1$ (flipping the chain left to right) which sends $f\rightarrow f^{-1}$. Define  $G_n^k=|0\rangle\langle 0|^{\otimes n}+|\chi^{k,n}\rangle\langle\chi^{k,n}|$ where
\begin{equation}
|\chi^{k,n}\rangle=\sqrt{1-|f|^2}\left(\sum_{i=1}^{k} f^{i-1}|\hat{0}^{\otimes i-1}\; \hat{1}\; \hat{0}^{\otimes n-i}\rangle\right) \qquad 1\leq k\leq n.
\label{eq:chibasis2}
\end{equation}
Noting that the ground space projector is $G_n=|0\rangle\langle 0|^{\otimes n}+\|\chi^{n,n}\|^{-2}|\chi^{n,n}\rangle\langle\chi^{n,n}|$ and that
\[
\left\||\chi^{n,n}\rangle\langle\chi^{n,n}|-|\chi^{k,n}\rangle\langle\chi^{k,n}|\right\|\leq 2\left\| |\chi^{n,n}\rangle -|\chi^{k,n}\rangle\right\|\leq 2\sqrt{1-|f|^2}\left(\sum_{i=k+1}^{\infty} |f|^{2(i-1)}\right)^{1/2}=2|f|^{k},
\]
we obtain
\[
\|G_n^{k}-G_n\|\leq 1-\|\chi^{n,n}\|^{2}+\left\||\chi^{n,n}\rangle\langle\chi^{n,n}|-|\chi^{k,n}\rangle\langle\chi^{k,n}|\right\|\leq |f|^{2n}+2|f|^k\leq 3|f|^k.
\]
Now let a partition $[n]=ABC$ be given with $|A|>|B|$ and $|C|=1$. Using the above bound three times, the triangle inequality, and the facts that $\|G_n^{k}\|\leq 1$ and $|f|^{|A|}<|f|^{|B|}$,  we obtain
\begin{equation}
\|G_{AB}G_{BC}-G_{ABC}\| \leq \|G^{|A|}_{AB}G^{|B|}_{BC}-G_{ABC}^{|A|}\|+9|f|^{|B|}.
\label{eq:gap1}
\end{equation}
Using explicit expressions for $G^{|A|}_{AB},G^{|B|}_{BC}$, and $G_{ABC}^{|A|}$ we see that $G^{|A|}_{AB}G^{|B|}_{BC}=G_{ABC}^{|A|}$ and the first term above is zero. This shows that the conditions of Lemma \ref{lem:gapcor} are satisfied with $\delta=|f|$, and $H_n(\Pi)$ is gapped.

Finally, we establish that $H_n(\Pi)$ is gapless if $|f|=1$. Using Eq.~\eqref{eq:Pi_orth2} we see that $H_n(\Pi)$ commutes with the total Hamming weight operator $\sum_{i=1}^{n} |\hat{1}\rangle\langle\hat{1}|_i$ and is therefore block diagonal with a block for each Hamming weight $0,\ldots,n$. The spectral gap $\gamma(\Pi,n)$ is upper bounded by the smallest nonzero eigenvalue within any given block. The matrix of the block with Hamming weight $1$, in the orthonormal basis
 $\{|e_i\rangle=f^{i-1}|\hat{0}^{i-1}\;\hat{1}\;\hat{0}^{n-i}\rangle : \; i=1,\ldots,n\}$, is given by
\[
\langle e_i|H_n(\Pi)|e_j\rangle=\frac{1}{2}\begin{cases} 1, \text{ if }i=j=1 \text{ or }i=j=n\\
2,  \text{ if }2\leq i=j\leq n-1,\\
-1,  \text{ if }  |i-j|=1\\
0,  \text{ otherwise}.\\
\end{cases}
\]
This is $1/2$ times the Laplacian matrix of the path graph of length $n$, and its spectrum is known. In particular, its smallest nonzero eigenvalue is  $\left(1-\cos(\pi/n)\right)$, which upper bounds $\gamma(\Pi,n)$ if $|f|=1$.
\qed
\end{document}